\documentclass{article}

\usepackage{amsmath,amsfonts}
\usepackage{amssymb}
\usepackage{amsthm}
\usepackage{xspace}
\usepackage{csquotes}
\usepackage{enumitem}
\usepackage{stmaryrd}
\usepackage{tabularx}
\usepackage{calc}

\newcommand{\EM}[1]{\ensuremath{#1}\xspace}
\newcommand{\newmathabb}[2]{\newcommand{#1}{\EM{{#2}}}}

\newcommand{\newrelation}[2]{\newmathabb{#1}{\mathtt{#2}}}
\newcommand{\newrelationann}[4]{\newmathabb{#1}{\mathtt{#2}^{#3}_{#4}}}
\newcommand{\newqueryclass}[3]{\newmathabb{#1}{\textbf{#2}^{#3}}}
\newcommand{\newabbrev}[2]{\newcommand{#1}{{#2}\xspace}}

\theoremstyle{definition}
\newtheorem{definition}{Definition}[section]
\newtheorem{proposition}[definition]{Proposition}

\newtheorem{example}[definition]{Example}

\newtheorem{theorem}[definition]{Theorem}
\newtheorem{lemma}[definition]{Lemma}
\newtheorem{corollary}[definition]{Corollary}
\newtheorem{claim}[definition]{Claim}
\newtheorem{myremark}[definition]{Remark}

\newcommand{\probdefbodynew}[3]{
\medskip
\noindent
\fbox{\parbox{\columnwidth}{
#1:\\
{\bf Input:} #2\\
{\bf Question:} #3 }}
\smallskip
}

\newcommand{\CharacterisationBox}[1]{\parbox{\linewidth-6em}{#1}}

\newcommand{\mydef}{\ensuremath{\mathrel{\smash{\stackrel{\scriptscriptstyle{
    \text{def}}}{=}}}}}

\newcommand{\IfDirection}{\smallskip\noindent\emph{(If)}\xspace}
\newcommand{\OnlyIfDirection}{\smallskip\noindent\emph{(Only if)}\xspace}

\renewcommand{\phi}{\varphi}

\newcommand{\cC}{{\cal C}}

\newabbrev{\naive}{na\"ive}

\newcommand{\newclass}[2]{\newmathabb{#1}{\EM{\textsc{#2}}}}

\newcommand{\polred}{\leq_p}

\newcommand{\containment}{\ensuremath{\text{{\sc Containment}}}}

\newcommand{\paracor}{\ensuremath{\text{{\sc Parallel-Correct}}}}
\newcommand{\paracom}{\ensuremath{\text{{\sc Parallel-Complete}}}}
\newcommand{\parasound}{\ensuremath{\text{{\sc Parallel-Sound}}}}
\newmathabb{\pitwoQBF}{\Pi_2\text{-QBF}}

\newcommand{\pitwo}{\EM{\Pi^p_2}}
\newcommand{\sigmatwo}{\EM{\Sigma^p_2}}

\newclass{\Poly}{P}
\newclass{\coNP}{coNP}
\newclass{\NP}{NP}

\newclass{\NEXPTIME}{{NEXPTIME}}
\newclass{\coNEXPTIME}{{coNEXPTIME}}
\newclass{\nexptime}{\NEXPTIME}
\newclass{\conexptime}{{\coNEXPTIME}}

\newcommand{\ssize}[1]{\EM{|#1|}} 

\newmathabb{\dom}{\mathbf{dom}} 
\newmathabb{\domk}{\mathit{dom}_k}
\newcommand{\adom}[1]{\EM{\mathit{adom}(#1)}}
\newcommand{\facts}[1]{\EM{\mathit{facts}(#1)}} 
\newcommand{\fcX}[1]{\EM{\boldsymbol{#1}}}

\newcommand{\restrict}[2]{\ensuremath{#1_{|#2}}}
\newmathabb{\const}{c}	
\newmathabb{\fc}{\fcX{f}} 
\newmathabb{\fcset}{\mathcal{F}}
\newmathabb{\fact}{\fcX{f}}
\newmathabb{\factB}{\fcX{g}}
\newmathabb{\fhead}{\fcX{h}} 
\newmathabb{\fcB}{\fcX{g}} 
\newmathabb{\qr}{\mathcal{Q}} 
\newmathabb{\sch}{\mathcal{D}} 
\newmathabb{\univ}{\mathbf{dom}} 
\newmathabb{\uvar}{\mathbf{var}} 
\newcommand{\ar}[1]{\EM{\mathit{ar}(#1)}} 
\newcommand{\tup}[2][]{\EM{\mathbf{#2}_{#1}}} 
\newmathabb{\atomset}{\mathcal{A}}
\newmathabb{\factset}{\mathcal{F}}

\newmathabb{\extend}{\mathit{ext}}
\newcommand{\layer}[1]{\ensuremath{\alpha_{#1}}}
\newcommand{\alphainv}[1]{\ensuremath{\alpha^{-1}_{#1}}}
\newcommand{\delayer}{\alphainv}

\newmathabb{\distp}{\boldsymbol{P}} 
\newmathabb{\nw}{\mathcal{N}}

\newmathabb{\fcC}{\fcX{h}}
\newcommand{\distwo}{\textit{loc-inst}}
\newcommand{\dist}[1]{\ensuremath{\textit{\distwo}_{\distp,#1}}}
\newcommand{\distvar}[2]{\ensuremath{\textit{\distwo}_{#1,#2}}}
\newmathabb{\node}{\kappa}	
\newmathabb{\nodeset}{N}	
\newmathabb{\cnodeA}{\node_1}	
\newmathabb{\cnodeB}{\node_2}		
\newcommand{\respwo}{\textit{rfacts}}
\newcommand{\respp}{\ensuremath{\respwo_\distp}}
\newcommand{\resp}[1]{\ensuremath{\respwo_\distp(#1)}}
\newcommand{\respvar}[2]{\ensuremath{\respwo_{#1}(#2)}}

\newcommand{\classpnondet}{\ensuremath{{\cal P}_{\text{npoly}}}}
\newmathabb{\classpnondetk}{\classp^k_{\text{npoly}}}
\newcommand{\classpnondetarg}[1]{\EM{\classp^{#1}_{\text{npoly}}}}
\newcommand{\Classpnondet}{\ensuremath{\mathfrak{P}_{\text{npoly}}}}
\newcommand{\Classp}{\ensuremath{\mathfrak{P}}}
\newmathabb{\classp}{{\cal P}}
\newmathabb{\classpnp}{{\cal P}_{\NP}}
\newmathabb{\classpfin}{{\cal P}_{\text{fin}}}
\newmathabb{\classpfinwild}{{\cal P}_{\text{fin}^*}} 
\newmathabb{\classprule}{\classp_{\text{rule}}}
\newmathabb{\classpenum}{\classprule}
\newmathabb{\classpenumexp}{\classp^{\text{\dom}}_{\text{rule}}}

\newmathabb{\drule}{\rho}

\newmathabb{\algNPDist}{\mathcal{A}_{\distp}}
\newmathabb{\addInput}{x}
\newmathabb{\addInputk}{x_1 \circ\dots\circ x_k}
\newmathabb{\addInputi}{x_i}

\newcommand{\one}[3]{\EM{[#1,#2](#3)}}   
 
\newcommand{\onePI}[1][\qr]{\one{#1}{\distp}{I}}

\newcommand{\body}[1]{\EM{\mathit{body}_{#1}}}

\newcommand{\pos}[1]{\ensuremath{\textit{pos}_{#1}}}
\newcommand{\negbody}[1]{\ensuremath{\textit{neg}_{#1}}}

\newcommand{\bodypos}[1]{\ensuremath{\textit{pos}_{#1}}}
\newcommand{\bodyneg}[1]{\ensuremath{\textit{neg}_{#1}}}

\newcommand{\head}[1]{\EM{\mathit{head}_{#1}}} 
\newcommand{\ineq}[1]{\EM{\mathit{ineq}_{#1}}}
\newcommand{\vars}[1]{\EM{\mathit{vars}(#1)}} 
\newcommand{\varmax}[1]{\EM{\mathit{varmax(#1)}}}

\newcommand{\CQ}{\ensuremath{{\text{CQ}}}\xspace}
\newcommand{\CQneg}{\ensuremath{{\text{CQ}}^{\lnot}}\xspace}

\newcommand{\CQnegineq}{\ensuremath{{\text{CQ}}^{\lnot,\neq}}\xspace}
\newmathabb{\UCQneg}{\text{UCQ}^\lnot} 

\newmathabb{\SCQneg}{\text{SCQ}^\lnot} 
\newmathabb{\FCQneg}{\text{FCQ}^\lnot}
\newmathabb{\SUCQneg}{\text{USCQ}^\lnot}
\newmathabb{\USCQneg}{\SUCQneg}
\newmathabb{\FUCQneg}{\text{UFCQ}^\lnot}
\newmathabb{\FUCQnegineq}{\text{UFCQ}^{\lnot,\neq}} 
\newmathabb{\FCQnegineq}{\text{FCQ}^{\lnot,\neq}} 
\newmathabb{\UCQineq}{\text{UCQ}^{\neq}} 
\newmathabb{\UCQineqs}{\ensuremath{\text{UCQ}^{\neq}\text{s}}} 
\newmathabb{\UCQnegineq}{\text{UCQ}^{\lnot,\neq}} 
\newmathabb{\UFCQneg}{\FUCQneg}
\newmathabb{\UFCQnegineq}{\FUCQnegineq}
\newcommand{\UCQ}{\ensuremath{{\text{UCQ}}}\xspace}
\newcommand{\UCQs}{\ensuremath{{\text{UCQs}}}\xspace}
\newcommand{\CQs}{\ensuremath{{\text{CQs}}}\xspace}

\newcommand{\CQnegineqs}{\ensuremath{{\text{CQ}}^{\lnot,\neq}}s\xspace}
\newcommand{\UCQnegineqs}{\ensuremath{{\text{UCQ}}^{\lnot,\neq}}s\xspace}
\newmathabb{\BCQ}{\text{BCQ}}
\newmathabb{\BCQneg}{\text{BCQ}^{\lnot}}
\newmathabb{\BCQineq}{\text{BCQ}^{\neq}}

\newqueryclass{\classCQ}{CQ}{}
\newqueryclass{\classCQneg}{CQ}{\lnot}
\newqueryclass{\classCQineq}{CQ}{\neq}
\newqueryclass{\classCQnegineq}{CQ}{\lnot,\neq}
\newqueryclass{\classUCQ}{UCQ}{}
\newqueryclass{\classUCQneg}{UCQ}{\lnot}
\newqueryclass{\classUCQineq}{UCQ}{\neq}
\newqueryclass{\classUCQnegineq}{UCQ}{\lnot,\neq}
\newqueryclass{\classSCQneg}{SCQ}{\lnot}
\newqueryclass{\classUSCQneg}{USCQ}{\lnot}
\newqueryclass{\classBCQ}{BCQ}{}
\newqueryclass{\classBCQneg}{BCQ}{\lnot}
\newqueryclass{\classBCQineq}{BCQ}{\neq}
\newqueryclass{\classUBCQneg}{UBCQ}{\lnot}
\newqueryclass{\classFCQneg}{FCQ}{\lnot}
\newqueryclass{\classFCQnegineq}{FCQ}{\lnot,\neq}
\newqueryclass{\classFUCQneg}{UFCQ}{\lnot}
\newqueryclass{\classUFCQneg}{UFCQ}{\lnot}
\newqueryclass{\classFUCQnegineq}{UFCQ}{\lnot,\neq}
\newqueryclass{\classUFCQnegineq}{UFCQ}{\lnot,\neq}

\newmathabb{\bw}{\boldsymbol{w}}
\newmathabb{\bx}{\boldsymbol{x}}
\newmathabb{\by}{\boldsymbol{y}}
\newmathabb{\nx}{\bar{x}}
\newmathabb{\ny}{\bar{y}}

\newmathabb{\VarTripPos}{\boldsymbol{W}^+}

\newrelation{\Global}{Global}
\newrelation{\True}{True}
\newrelation{\False}{False}
\newrelation{\Active}{Active}
\newrelation{\Var}{Var}
\newrelation{\Dom}{Bool}
\newrelation{\Neg}{Neg}
\newrelation{\Clause}{Clause}
\newrelation{\Succ}{Succ}
\newrelation{\Rel}{Rel}
\newrelation{\Type}{Type}
\newrelationann{\StartA}{Start}{}{1}
\newrelationann{\StartB}{Start}{}{2}
\newrelation{\Stop}{Stop}

\newmathabb{\qrglobal}{\qr_\Global}
\newmathabb{\qrnotglobal}{\qr_{\lnot\Global}}
\newmathabb{\qrBglobal}{\qr'_\Global}
\newmathabb{\qrBnotglobal}{\qr'_{\lnot\Global}}

\newmathabb{\IV}{I_V}
\newmathabb{\IW}{I_W}

\newmathabb{\AtomsReq}{\atomset_\sch}
\newmathabb{\AtomsPro}{\atomset^\lnot_\sch}
\newmathabb{\AtomsActive}{\atomset_{\text{act}}}
\newmathabb{\AtomsPos}{\atomset_{\text{pos}}}
\newmathabb{\AtomsCons}{\atomset_{\text{cons}}}
\newmathabb{\AtomsVar}{\atomset_{\text{var}}}
\newmathabb{\AtomsVarX}{\atomset^\bx_{\text{var}}}
\newmathabb{\AtomsVarY}{\atomset^\by_{\text{var}}}
\newmathabb{\AtomsSat}{\atomset_{\text{sat}}}

\newmathabb{\FactsPos}{\factset_{\text{pos}}}
\newmathabb{\FactsCons}{\factset_{\text{cons}}}
\newmathabb{\FactsVarX}{\factset^\bx_{\text{var}}}

\newmathabb{\Vdef}{V_{\text{def}}}
\newmathabb{\InstReq}{J_\sch}
\newmathabb{\InstPro}{J^\lnot_\sch}

\newmathabb{\nodeDef}{\node_{\text{def}}}
\newmathabb{\nodeTrue}{\node_\True}
\newmathabb{\nodeFalse}{\node_\False}

\newmathabb{\nodeComp}{\kappa}
\newmathabb{\nodeSound}{\sigma}
\newmathabb{\nodeRest}{\rho}

\newmathabb{\ConditionPC}{\mathit{C1}}
\newmathabb{\ConditionPCUCQ}{\mathit{C1'}}

\title{Parallel-Correctness and Containment for Conjunctive Queries with Union and Negation}
\author{Gaetano Geck \\ {\small TU Dortmund University} 
 \and Bas Ketsman \\ {\small Hasselt University} \and Frank Neven \\ {\small Hasselt University}
\and Thomas Schwentick \\ {\small TU Dortmund University}
}

\date{}

\begin{document}

\maketitle

\maketitle

\begin{abstract}
Single-round multiway join algorithms first reshuffle data over many servers and then evaluate the query at hand in a parallel and communication-free way. 
A key question is whether a given distribution policy for the
reshuffle is adequate for computing a given query, also referred to
as parallel-correctness. This paper extends the study of the
complexity of parallel-correctness and its constituents,
parallel-soundness and parallel-completeness, to unions of conjunctive
queries with and without negation. As a by-product it is shown that
the containment problem for conjunctive queries with negation is
\conexptime-complete.
\end{abstract}

\section{Introduction}
\label{sec:intro}

Motivated by recent in-memory systems like Spark~\cite{spark} and
Shark~\cite{shark}, Koutris and Suciu introduced the massively parallel
communication model (MPC)~\cite{DBLP:conf/pods/KoutrisS11} where
computation proceeds in a sequence of parallel steps each followed by
global synchronisation of all servers. Of particular interest in the
MPC model are queries that can be evaluated in one round of
communication~\cite{DBLP:conf/pods/BeameKS14}. In its most \naive
setting, a query $\qr$ is evaluated by reshuffling the data over many
servers, according to some distribution policy,  and then computing $\qr$ at each server in a parallel but communication-free manner. 
A notable family of distribution policies is formed within the Hypercube algorithm~\cite{AfratiUllman10,DBLP:conf/pods/BeameKS14,DBLP:conf/sigmod/ChuBS15}. 
 A property of Hypercube distributions is that for any instance $I$, the central execution of $\qr(I)$ always equals the union of the evaluations of $\qr$ at every computing node (or server). The latter guarantees the correctness of the distributed evaluation for any conjunctive query by the
 Hypercube algorithm. 

 Ameloot et al.~\cite{DBLP:conf/pods/AmelootGKNS15} introduced a general framework for reasoning about one-round evaluation algorithms under {\em arbitrary} distribution policies. They introduced {\em parallel-correctness} as a property of a query w.r.t.\ a distribution policy which states that central execution always equals distributed execution, that is, equals
 the union of the evaluations of the query at each server under the given distribution policy. One of the main results
 of \cite{DBLP:conf/pods/AmelootGKNS15} is that deciding parallel-correctness
 for conjunctive queries (\CQ{}s) is $\Pi_2^P$-complete under arbitrary distribution policies.  The upper bound follows rather directly from a semantical characterisation of parallel-correctness in terms of properties of minimal valuations. Specifically, it was shown that a conjunctive query is parallel-correct w.r.t.\ a distribution policy, if the distribution policy sends for every minimal valuation its required facts to at least one node.
 
As union and negation are fundamental operators, we extend in this paper the study of parallel-correctness to unions of conjunctive queries (\UCQ), conjunctive queries with negation (\CQneg) and unions of conjunctive queries with negation (\UCQneg). In fact, we study two addional but related notions: parallel-soundness and parallel-completeness. While parallel-correctness implies equivalence between
 centralised and distributed execution, 
 parallel-soundness (respectively, parallel-completeness) requires
 that distributed execution is contained in (respectively, contains)
  centralised execution.
  Of course, parallel-soundness and
  parallel-completeness together are equivalent to parallel-correctness. Furthermore, since all monotone queries are parallel-sound, on this class parallel-correctness is equivalent to parallel-completeness.

 We start by investigating parallel-correctness for 
 \UCQ. Interestingly, for a \UCQ to be parallel-correct under a certain distribution policy it is not required that
 every disjunct is parallel-correct.  We extend the characterisation for parallel-correctness in terms of
 minimal valuations for \CQs to \UCQs and thereby obtain membership in $\Pi_2^P$. The matching lower bound follows, of course,  from the lower bound for \CQs~\cite{DBLP:conf/pods/AmelootGKNS15}.

   Next, we study parallel-correctness for (unions of) conjunctive
   queries with negation. Sadly, when negation comes into play,
   parallel-correctness can no longer be characterised in terms of
   properties of valuations. Instead our algorithms are based on counter-examples of exponential size, yielding \conexptime upper bounds. It turns out that this is optimal, though, as our corresponding lower bounds show.  The proof of the lower bounds comes along an unexpected route:
 we exhibit a reduction from query
   containment for $\CQneg$ to parallel-correctness of \CQneg (and its
   two variants) and show that query containment
   for $\CQneg$ is \conexptime-complete. This is considerably different from what we
   thought was folklore knowledge of the community. Indeed,
   the $\Pi_2^p$-completeness result for query containment for
   $\CQneg$ mentioned in \cite{DBLP:journals/tcs/Ullman00} only seems
   to hold for fixed database schemas (or a fixed arity
   bound, for that matter). We note that Mugnier et al.~\cite{MugnierST12}
   provide a $\pitwo$ upper bound proof for $\CQneg$ containment and explicitly mention that it holds under the assumption that
   the arity of predicates is bounded by a constant. 
Altogether, parallel-correctness (and its
   variants) for (unions of) conjunctive queries with negation is thus complete for \conexptime. 
 
Finally, a natural question is how the high complexity of
  parallel-correctness in the presence of negation can be lowered. We identify
  two 
  cases in which the complexity drops. 
  More specifically, 
  the complexity decreases from \conexptime to
 \pitwo if the database schema is fixed or the arity of relations is bounded, and to \coNP for unions of \emph{full} conjunctive queries with negation. In the latter case, we again employ a reduction from containment of full conjunctive queries (with negation) and obtain novel results on the containment problem in this setting as well.
   All upper bounds 
   hold
   for queries with inequalities.

\bigskip
\noindent
{\bf Outline.} This paper is further organised as follows.
In Section~\ref{sec:relwork}, we discuss related work. 
In Section~\ref{sec:defs}, we introduce the necessary definitions.
We address parallel-correctness for unions of conjunctive queries in Section~\ref{sec:ucq}. We consider containment of conjunctive queries with negation in Section~\ref{sec:containment} and parallel-correctness together with its variants in Section~\ref{sec:negation}.
We discuss the restriction to full conjunctive queries in Section~\ref{sec:full}.
We conclude in Section~\ref{sec:discussion}.

\section{Related work}
\label{sec:relwork}

As mentioned in the introduction, Koutris and Suciu introduced the massively
parallel communication model (MPC)~\cite{DBLP:conf/pods/KoutrisS11}. A key property is that computation proceeds in a sequence of parallel steps, each followed by global synchronisation of all computing nodes. In this model, evaluation of conjunctive queries~\cite{DBLP:conf/pods/BeameKS13,DBLP:conf/pods/KoutrisS11} and skyline queries~\cite{DBLP:conf/icdt/AfratiKSU12} has been considered. Beame, Koutris and Suciu~\cite{DBLP:conf/pods/BeameKS14} proved a matching upper and lower bound for the amount of communication needed to compute a full conjunctive query without self-joins in one communication round. The upper bound is provided by a randomised algorithm called \emph{Hypercube} which uses a technique that can be traced back to Ganguly, Silberschatz, and Tsur~\cite{DBLP:journals/jlp/GangulyST92} and is described in the context of map-reduce by Afrati and Ullman~\cite{AfratiUllman10}.

Ameloot et al.~\cite{DBLP:conf/pods/AmelootGKNS15} introduced a general framework for reasoning about one-round evaluation algorithms under {\em arbitrary} distribution policies. They introduced the notion of {\em parallel-correctness} and proved its associated decision problem to be $\pitwo$-complete for conjunctive queries. In addition, towards optimisation in MPC, they considered parallel-correctness \emph{transfer}. Here, parallel-correctness transfers from $\qr$ to $\qr'$ when $\qr'$ is parallel-correct under every distribution policy for which $\qr$ is parallel-correct. The associated decision problem for conjunctive queries is shown to be $\Pi_3^p$-complete. In addition, some restricted cases (e.g., transferability under Hypercube distributions), are shown to be \NP-complete.

Our definition of a distribution policy is borrowed from Ameloot et
al.~\cite{DBLP:conf/pods/AmelootKNZ14} (but already surfaces in the work of Zinn et
al.~\cite{DBLP:conf/icdt/ZinnGL12}), where distribution policies are used to define the
class of policy-aware transducer networks. 
The work by Ameloot et al.~\cite{DBLP:journals/jacm/AmelootNB13,DBLP:conf/pods/AmelootKNZ14} relates coordination-free computation with definability in variants of Datalog. One-round communication algorithms in MPC can be seen as
very restrictive coordination-free computation.

The complexity of query containment for conjunctive queries is proved to be \NP-complete by Chandra and Merlin~\cite{DBLP:conf/stoc/ChandraM77}. Levy and Sagiv provide a test for query containment of conjunctive queries with negation~\cite{DBLP:conf/vldb/LevyS93} that involves exploring an exponential
number of possible counter-example instances. In the context of information
integration, Ullman~\cite{DBLP:journals/tcs/Ullman00} gives a comprehensive overview of query containment (with and without negation) and states the complexity of query containment for \CQneg to be $\Pi_2^p$-complete. As mentioned in the introduction, the latter apparently only holds when the database schema is fixed
or the arity of relations is considered to be bounded.
A proof for the $\Pi_2^p$-lowerbound is given by
Farr\'e et al. \cite{DBLP:conf/icdt/FarreNTU07}.
Based on \cite{DBLP:conf/vldb/LevyS93}, Wei and Lausen~\cite{DBLP:conf/icdt/WeiL03} study a method for testing
containment that exploits containment mappings for the positive parts of queries, and additionally provide a characterisation for \UCQneg{} containment.

\section{Definitions}
\label{sec:defs}

\subsection{Queries and instances}
We assume an infinite set \dom of data values that can be represented
by strings over some fixed alphabet. By $\dom_n$ we denote the set of
data values represented by strings of length at most $n$.
A \emph{database schema} $\sch$ is a finite set of relation names $R$,
each with some arity $\ar R$. We also write $R^{(k)}$ as a shorthand to denote that $R$ is a relation of arity $k$. 
We call $R(\tup t)$ a \emph{fact} when $R$ is a relation name and
$\tup t$ a tuple over \dom of appropriate arity.
We say that a fact $R(\tup t)$ is {\em over} a database
schema $\sch$ if $R \in \sch$. For a subset $U\subseteq \dom$ we write 
$\facts{\sch,U}$ for the set of possible facts over schema $\sch$ and
$U$ and by $\facts \sch$ we denote $\facts{\sch,\dom}$.
A \emph{(database) instance} $I$ over $\sch$ is a finite set of facts
over $\sch$.  By $\adom{I}$ we denote the set of data values occurring
in $I$. 
A \emph{query $\qr$ over input schema $\sch_1$ and output schema
  $\sch_2$} is a generic mapping from instances over $\sch_1$ to
instances over $\sch_2$.  Genericity means that for every permutation
$\pi$ of \dom\ and every instance $I$, $\qr(\pi(I)) =
\pi(\qr(I))$.
We say that $\qr$ is {\em contained} in $\qr'$, denoted $\qr\subseteq\qr'$ iff
for all instances $I$, $\qr(I)\subseteq \qr'(I)$.

\subsection{Unions of conjunctive queries with negation}
\label{sec:cq}

Let $\uvar$ be an infinite set of variables, disjoint from $\dom$.
An \emph{atom} over schema \sch is of the form $R(\tup{x})$, where $R$
is a relation name from \sch and $\tup{x}=(x_1,\ldots,x_k)$ is a tuple
of variables in $\uvar$ with $k = \ar R$. 
A \emph{conjunctive query $\qr$ with negation and inequalities} over input schema $\sch$ is an expression of the form
$$T(\tup{x}) \leftarrow R_1(\tup{y_1}), \ldots, R_m(\tup{y_m}), \lnot
S_1(\tup{z_1}), \ldots, \lnot S_n(\tup{z_n}),\beta_1,\ldots,\beta_p$$ 
where all $R_i(\tup[i] y)$ and $S_i(\tup[j] z)$ are atoms over $\sch$, every $\beta_i$ is
an inequality of the form $s\neq s'$ where $s,s'$ are distinct variables occurring in some $\tup[i]
y$ or $\tup[j] z$, and $T(\tup x)$ is an atom for which $T \not \in \sch$. Additionally, for safety, we require that every variable in $\tup x$
occurs in some $\tup[i] y$
and that every variable occurring in a
negated atom has to occur in a positive atom as well 
(\emph{safe} negation). We refer to the \emph{head atom} $T(\tup
x)$ as $\head{\qr}$, to the set $\{R_1(\tup{y_1}), \ldots,
R_m(\tup{y_m}),S_1(\tup{z_1}), \ldots,$ $S_n(\tup{z_n})\}$ as
$\body{\qr}$, and to the set $\{\beta_1,\dots,\beta_p\}$ as $\ineq{\qr}$.
Specifically, we refer to $\{R_1(\tup{y_1}), \ldots, R_m(\tup{y_m})\}$ as the positive atoms in $\qr$, denoted $\pos{\qr}$, and to $\{S_1(\tup{z_1}), \ldots, S_n(\tup{z_n})\}$ as the \emph{negated} atoms of $\qr$, denoted $\negbody{\qr}$. We denote by $\vars{\qr}$ the set of all variables occurring in $\qr$.
We refer to the class of conjunctive queries with negation and inequalities by
\classCQnegineq, its restriction to queries without inequalities,
without negated atoms, and without both by \classCQneg,
\classCQineq, and \classCQ, respectively. As a shorthand we refer to queries from
\classCQnegineq as \CQnegineqs and similarly for the other classes.

A \emph{pre-valuation} for a \CQnegineq~\qr is a total function
$V:\vars{\qr} \to \dom$, which naturally extends to atoms and sets of atoms.
It is {\em consistent} for $\qr$, if $V(\pos{\qr})
\cap V(\negbody{\qr}) = \emptyset$, and
$V(s)\not=V(s')$, for every inequality $s\not=s'$ of \qr, in which case it is called a valuation.
Of course, for a conjunctive query without negated atoms and without inequalities, every pre-valuation is also a valuation.
We refer to $V(\pos{\qr})$ as the facts \emph{required} by $V$, and to $V(\negbody{\qr})$ as the facts
\emph{prohibited} by $V$.

A valuation $V$ \emph{satisfies} $\qr$ on instance $I$ if all facts required by $V$ are in $I$ while no fact prohibited by $V$ is in $I$, that is, if $V(\bodypos{\qr}) \subseteq I$ and $V(\bodyneg{\qr}) \cap I = \emptyset$.
In that case, $V$ \emph{derives} the fact $V(\head{\qr})$. The \emph{result
of $\qr$ on instance $I$}, denoted $\qr(I)$, is defined as the set of
facts that can be derived by satisfying valuations for $\qr$ on
$I$.

A {\em union of conjunctive queries with negation and inequalities} 
is a finite union of \CQnegineqs. That is, $\qr$ is of the form $\bigcup^n_{i=1}
\qr_i$ where all subqueries $\qr_1,\dots,\qr_n$ have the same relation name in their head atoms.
We assume disjoint variable sets among different disjuncts in 
$\qr$. That is, $\vars{\qr_i} \cap \vars{\qr_j}=\emptyset$ for $i\neq j$
and, in particular, $\vars{\head{\qr_i}}\neq \vars{\head{\qr_j}}$. 
By $\varmax{\qr}$ we denote the maximum number of variables that
occurs in any disjunct of $\qr$. 
By
\classUCQnegineq we denote the class of unions of conjunctive
queries with negation and inequalities and its fragments are denoted correspondingly. 

A \CQnegineq is called \emph{full} if all of its variables occur in its head. A \UCQnegineq is \emph{full} if all its subqueries are full.

The \emph{result of $\qr$ on instance $I$} is
$\qr(I)=\bigcup_{i=1}^n \qr_i(I)$. Accordingly, a mapping from variables to data values is a \emph{valuation} for a
\UCQnegineq~$\qr$ if it is a valuation for one of its
subqueries.

\subsection{Networks, data distribution, and policies}

A \emph{network} $\nw$ is a nonempty finite set of values from \dom,
which we call \emph{(computing) nodes} (or servers).
A \emph{distribution policy} $\distp=(U,\respp)$ for a database schema $\sch$ and
a network $\nw$ consists of a universe $U$ and  a total function
$\respp$ that maps each node of $\nw$ to a
set of facts from  $\facts{\sch,U}$.
A node~$\node$ is \emph{responsible for fact~$\fc$} (under policy~$\distp$) if $\fc \in \respp(\node)$.
As a shorthand (and slight abuse
of notation), we denote the set
of nodes $\kappa$ that are responsible for some given fact \fc by $\distp(\fc)$. 
For a distribution policy $\distp$ and an instance $I$ over $\sch$, let
$\dist I$ denote the function that maps each $\node\in\nw$ to $I\cap \resp{\node}$,
that is, the set of facts in $I$ for which $\node$ is responsible. 
We sometimes refer to a given instance~$I$ as the \emph{global instance} and to
$\dist{I}(\node)$ as the \emph{local instance at node~$\node$}. 

We note that for some facts from $\facts{\sch,U}$ there are no
responsible nodes. This gives our framework some additional
flexibility. However, it does not affect our results: in the lower bound proofs we only use
distributions for which all facts from $\facts{\sch,U}$ have some
responsible nodes.
Each distribution policy implicitly induces a network and
each query implicitly defines a database (sub-) schema. Therefore, we often
omit the explicit notation for networks and schemas.

Given some policy~$\distp$ that is defined over a network~$\nw$, the \emph{result $\onePI$ of the distributed evaluation of a query~$\qr$ on an instance~$I$ in one round} is defined as the union of the results of the query evaluated on each node's local instance. Formally,
\[
\onePI\mydef\bigcup_{\node\in\nw}{\qr\big(\dist I(\node)\big)}.
\]

In the decision problem for parallel
correctness (to be formalised later), the input consists of a query
$\qr$ and a distribution policy $\distp$. However, it is not obvious
how distribution policies should be specified. In principle, they
could be defined in an arbitrary fashion, but it is reasonable to
assume that given a potential fact \fc, a node \node and a policy \distp, it is
not too hard to find out whether \node is responsible for \fc under
\distp.

For \UCQineqs, which are monotone,
our complexity results are remarkably robust with respect to
the choice of the representation of distribution policies. In fact, the complexity results coincide for the two extreme
possible choices that we consider in this article. In the first case,  distribution policies are specified by an explicit list of tuple-node-pairs,
whereas in the second case the test whether a given node is responsible for a
given tuple can be carried out by a non-deterministic polynomial-time algorithm. However, we do require that some bound $n$ on the length of strings that represent node names and data values is given. Without
such a restriction, no upper complexity bounds would be possible as nodes
with names of super-polynomial length in the size of the input 
would not be accessible.

Considering queries with negated
atoms, however, these two settings (seem to) differ, complexity-wise.
Thus, we consider a third option, $\classpenum$, in which the universe $U$ of a policy is
explicitly enumerated and the responsibilities are defined by simple
constraints (described below). The latter representation enjoys the same
complexity properties as the full \NP-test based case.

Now we give more precise definitions of classes of policies and 
their representations as inputs of algorithmic problems.
As said before, policies $\distp=(U,\respp)$ from $\classpfin$ are specified by an
explicit enumeration of $U$ and of all pairs $(\node,\fc)$  where
$\node\in\distp(\fc)$.
A policy $\distp=(U,\respp)$ from $\classpenum$ is given by an
explicit enumeration of $U$ and a list of \emph{rules} of the form
$\drule=(A, \kappa)$, where $A$ is an atom with variables and/or constants from $U$, and a network node $\kappa$.
The semantics of such a rule is as follows: for
every substitution $\mu: \uvar \cup \dom \to
\dom$ that maps variables to values from $U$ and leaves constants
from $U$ unchanged, the node \node is responsible for the fact $\mu(A)$.
A rule is a \emph{fact rule} if its atom does not
contain any variables, that is, $A = R(a_1,\dots,a_n)$, where
$a_1,\dots,a_n \in U$. In particular, $\classpfin \subseteq \classpenum$.

\begin{example}
	Let distribution policy~$\distp$ over schema~$\{\Rel^{(3)}\}$ and network
	$\{\cnodeA,\cnodeB\}$ be given by $U=\{1,\ldots,10\}$ and the
        rules 
	$\big(\Rel(1,x,x), \cnodeA\big), \big(\Rel(2,x,y), \cnodeB\big)$. On global
	instance $I=\{\Rel(1,7,7), \Rel(1,7,8), \Rel(2,9,8), \Rel(2,9,9)
	\}$, policy $\distp$ induces local instances $\dist{I}(\cnodeA) =
	\{\Rel(1,7,7)\}$ and $\dist{I}(\cnodeB) = \{\Rel(2,9,8),
	\Rel(2,9,9)\}$. \hfill $\Box$
\end{example}

The most general classes of policies allow to specify policies by
means of a \enquote{test algorithm} with time bound $\ell^k$, where $\ell$ is the length of the input and
$k$ some constant. Such an algorithm decides, for an input consisting of a
node~$\node$ 
and fact~$\fc$, whether $\node$ is responsible for
$\fc$.%
\footnote{We note that it is important that for each
class of policies there is a fixed $k$ that bounds the exponent in the
test algorithm as otherwise we could not expect a polynomial bound for
all policies of that class.}
 A policy $\distp=(U,\respp)$ from $\classpnondetk$
is specified by a pair $(n,\algNPDist)$, where $n$ is a natural number in unary representation and $\algNPDist$ is a non-deterministic algorithm.\footnote{For
concreteness, say, a non-deterministic Turing machine.} The universe $U$ of
$\distp$ is the set of all data values that can be represented by strings of
length at most $n$ (for some given fixed alphabet) and the underlying network
consists of all nodes which are represented by strings of length at most $n$, that is, $\nw=\dom_n$. A
node \node is responsible for a fact \fc if $\algNPDist$, on input
$(\node,\fc)$, has an accepting run of at most $|(\node,\fc)|^k$ steps. Clearly, each
policy of $\classpfin$ can be described in
$\classpnondetarg{2}$. Let $\Classpnondet$ denote the set\footnote{Since \enquote{linear time} is
a subtle notion, we rather not consider
$\classpnondetarg{1}$.}
 $\{\classpnondetk
\mid k\ge 2\}$ of distribution policies  and by $\Classp$ the set $\{\classpfin,\classpenum\}\cup
\Classpnondet$.

\subsection{Parallel-correctness, soundness, and completeness}
\label{sec:notions}

In this paper, we mainly consider the one-round evaluation
algorithm for a query \qr that first distributes (reshuffles) the data
over the computing nodes according to $\distp$, then evaluates $Q$ in
a parallel step at every computing node, and finally
outputs all facts that are obtained in this way.\footnote{We note
  that, since $\distp$ is defined on the granularity of a fact, the reshuffling does not depend on the current distribution of the data and can be done in parallel as well.}  
As formalised next, the one-round evaluation algorithm is correct (sound,
complete) if the query \qr is parallel-correct (parallel-sound, parallel-complete) under $\distp$.

\begin{definition}
    \label{def:pc-inst}
Let \qr be a query, $I$ an instance, and $\distp$ a distribution policy.
\begin{itemize}
	\item \qr is \emph{parallel-sound on $I$ under \distp}
	if $\qr(I) \supseteq \onePI$.
	\item \qr is \emph{parallel-complete on $I$ under \distp}
	if $\qr(I) \subseteq \onePI$; and,
	\item \qr is \emph{parallel-correct on $I$ under \distp} if
	$\qr(I) = \onePI$, that is,\\
	if it is parallel-sound and parallel-complete.
\end{itemize}
\end{definition}

\begin{definition}
\label{def:pc:general}
A query $\qr$ is \emph{parallel-correct (respectively, parallel-sound and parallel-complete) under distribution policy} $\distp=(U,\respp)$, if $\qr$ 
is parallel-correct (respectively, parallel-sound and
parallel-complete) on all instances $I\subseteq\facts{\sch,U}$.
\end{definition}

In \cite{DBLP:conf/pods/AmelootGKNS15}, parallel-correctness is characterised
in terms of minimal valuations as defined next:
\begin{definition}\label{def:minimal:cq}
Let $\qr$ be a CQ. A valuation $V$ for $\qr$ is \emph{minimal} for $\qr$ if there exists no valuation $V'$ for $\qr$ such that $V(\head{\qr})=V'(\head{\qr})$ and $V'(\body{\qr}) \subsetneq V(\body{\qr})$.
\end{definition}

The following lemma is key in obtaining the $\pitwo$ upper bound on
the complexity of testing parallel-correctness for conjunctive
queries:
\begin{lemma}[Characterisation of parallel-correctness
	for $\CQ$s~\cite{DBLP:conf/pods/AmelootGKNS15}]
    A \CQ~$\qr$ is parallel-correct under distribution
    policy $\distp=(U,\respp)$ if and only if the following holds:
	\begin{equation}
	        \tag{\ConditionPC}
	        \label{lem:pc:pods}
	        \CharacterisationBox{
	        	For every minimal valuation $V$ for \qr over $U$,
			there is a node $\node \in \nw$ such that
			\begin{math}
				V(\body{\qr}) \subseteq \respp(\node).
			\end{math}
		}
	\end{equation}
\end{lemma}

\begin{myremark}\label{remark:C1} 
Informally, condition (C1) states that there is a node in the network where all facts required for $V$ meet. 
\end{myremark}

\subsection{Algorithmic problems}
\label{sec:def:decproblems}

We consider the following decision
problems for various sub-classes  $\cC$ and $\cC'$ of  \classUCQnegineq and classes
\classp of distribution policies from $\{\classpfin, \classpenum\}\cup\Classpnondet$.

\begin{center}
\begin{minipage}{0.48\linewidth}
	\probdefbodynew{$\containment(\cC,\cC')$}
		{$\qr \in \cC$ and $\qr' \in \cC'$}
		{Is $\qr \subseteq \qr'$?}
\end{minipage}~\hspace*{0.5em}
\begin{minipage}{0.47\linewidth}
	\probdefbodynew{$\parasound(\cC,\classp)$}
		{$\qr \in \cC$, $\distp \in \classp$}
		{Is $\qr$ parallel-sound under $\distp$?}
\end{minipage}\hfill~
\end{center}
\vspace{-3mm}
\begin{center}
\begin{minipage}{0.48\linewidth}
	\probdefbodynew{$\paracom(\cC,\classp)$}
		{$\qr \in \cC$, $\distp \in \classp$}
		{Is $\qr$ parallel-complete under $\distp$?}
\end{minipage}~\hspace*{0.5em}
\begin{minipage}{0.47\linewidth}
	\probdefbodynew{$\paracor(\cC,\classp)$}
		{$\qr \in \cC$, $\distp \in \classp$}
		{Is $\qr$ parallel-correct under $\distp$?}
\end{minipage}\hfill~
\end{center}

\section{Parallel-correctness: unions of conjunctive queries}
\label{sec:ucq}

Parallel-correctness of unions of conjunctive queries (without
negation) reduces to parallel-completeness for the simple reason that
 these queries are monotone and therefore
parallel-sound for every distribution policy. We show below that
parallel-completeness remains in $\pitwo$. Hardness already 
follows from $\pitwo$-hardness of $\paracor(\classCQ,\classpfin)$~\cite{DBLP:conf/pods/AmelootGKNS15}.

As a $\UCQ$ is 
parallel-complete under a policy $\distp$ when all its disjuncts are,
it might be tempting to assume that this condition is also necessary. However, as the following example illustrates, this is not the case.

\begin{example}
	\label{ex:complete-ucq-incomplete-disjunct}
	Let $\qr=\qr_1\cup\qr_2$, where $\qr_1$ and $\qr_2$ are the
        following \CQs:
	\begin{displaymath}
		\begin{array}{llcl}
			\qr_1: & H(x,x) & \gets & R(x,x), \\
			\qr_2: & H(y,z) & \gets & R(y,z), S(y,z).
		\end{array}
	\end{displaymath}
	Further, let $\distp$ be the policy over network $\{\cnodeA,\cnodeB\}$ that
	maps facts $R(a,a)$ to node~$\cnodeA$, for all $a \in \dom$, and all other
	$R$-facts and all $S$-facts to node~$\cnodeB$.
	
	We argue that $\qr$ is parallel-complete under $\distp$ on all instances. Indeed, assume $H(a,b)\in \qr(I)$ for some instance $I$ and
	$a,b \in \dom$. If $a \neq b$, only the valuation $\{y \mapsto a, z
	\mapsto b\}$ can derive $H(a,b)$. This means that $\{R(a,b),
	S(a,b)\}\subseteq I$. Furthermore, $\{R(a,b),
	S(a,b)\}\subseteq \respwo_\distp(\cnodeB)$. Hence, 
	$H(a,b)\in \qr(\dist I(\cnodeB))$. 
	If $a = b$, then $R(a,a)\in I$. So, $R(a,a)\in \respwo_\distp(\cnodeA)$ and $H(a,a)\in \qr(\dist I(\cnodeA))$.
 On the other hand, $\qr_2$ is 
	not parallel-complete under $\distp$ on instance $I=\{R(0,0),S(0,0)\}$.
	Indeed, $H(0,0) \in \qr_2(I)$ but $\qr_2\big(\dist{I}(\cnodeA)\big) =
	\qr_2\big(\{R(0,0)\}\big) = \emptyset$ and $\qr_2\big(\dist{I}(\cnodeB)\big) =
	\qr_2\big(\{S(0,0)\}\big) = \emptyset$. 
	\hfill $\Box$
\end{example}

We recall from Section~\ref{sec:cq} that disjuncts in unions of
conjunctive queries use disjoint
variable sets and a valuation for $\qr$ is a valuation for exactly one disjunct.
As formalised next, the notion of minimality for valuations given in Definition~\ref{def:minimal:cq} naturally extends to \classUCQineq.

\begin{definition}
	Let $\qr=\bigcup^n_{i=1} \qr_i$ be a $\UCQineq$. A
       valuation~$V_i$ for $\qr_i$, with $i\in\{1,\dots,n\}$,
	is \emph{minimal} for $\qr$, if for no $j\in\{1,\dots,n\}$
        there is a valuation $V_j$ for $\qr_j$, such that
	$V_j(\head{\qr_j})=V_i(\head{\qr_i})$ and $V_j(\body{\qr_j}) \subsetneq
	V_i(\body{\qr_i})$.
\end{definition}

\begin{example}
	Consider a simple $\UCQineq$ $\qr=\qr_1\cup\qr_2$ where  $\qr_1,\qr_2\in
	\classCQineq$ are as follows:
	\begin{displaymath}
		\begin{array}{llcl}
			\qr_1: & H(u,v) & \gets &
			R(u,v), R(v,u), R(u,u), \\
			\qr_2: & H(x,y) & \gets &
			R(x,y), R(y,z), y \neq z. \\
		\end{array}
	\end{displaymath}
	Valuation $V_2 \mydef \{x \mapsto 0, y \mapsto 0, z \mapsto 1\}$ is not minimal
	for $\qr$ because valuation $V_1 \mydef \{u \mapsto 0, v \mapsto 0\}$ derives
	the same fact $H(0,0)$ requiring only $\{R(0,0)\}\subsetneq\{R(0,0),R(0,1)\}$. 	
	Similarly, valuation $W_1 \mydef \{u \mapsto 0, v \mapsto 1\}$, requiring 
	$\{R(0,1), R(1,0), R(0,0)\}$, is not minimal for $\qr$ because valuation $W_2
	\mydef \{x \mapsto 0, y \mapsto 1, z \mapsto 0\}$ only requires $\{R(0,1),
	R(1,0)\}$. \hfill $\Box$
\end{example}

The
notion of minimality leads to basically the same simple characterisation of
parallel-completeness:
\begin{lemma}\label{lem:pc:ucq}
    A \UCQineq~$\qr$ is parallel-correct under distribution
    policy $\distp=(U,\respp)$ if and only if the following holds:
	\begin{equation}
	        \tag{\ConditionPCUCQ}
	        \label{lem:pc:pods} 
	        \CharacterisationBox{
	        	For every minimal valuation $V$ for \qr over $U$,
			there is a node $\node \in \nw$ such that
			\begin{math}
				V(\body{\qr}) \subseteq \respp(\node).
			\end{math}
		}
	\end{equation}
\end{lemma}

\begin{proof}
	\IfDirection
	Assume $(C1')$ holds. Because of monotonicity, we only need to show that
	$\qr(I) \subseteq \bigcup_{\node\in\nw}{\qr(\dist I(\node))}$ for
	every instance $I$. To this end, let $\fc$ be an arbitrary fact that is derived
	by some valuation $V$ for $\qr$ on $I$. Then, there is also a minimal
	valuation $V'$ that is satisfying on $I$ and which derives $\fc$.
	Because of $(C1')$, there is a node $\node \in \nw$ where all facts required
	by $V'$ meet (cf.\ Remark~\ref{remark:C1}). Hence, $\fc\in \bigcup_{\node\in\nw}{\qr(\dist I(\node))}$, i.e.
	query~$\qr$ is parallel-correct under policy~$\distp$.

    \OnlyIfDirection
    For a proof by contraposition, suppose that there is a minimal
    valuation $V'$ for $\qr$ for which the required facts do \emph{not} meet under
    \distp.
    Consider the input instance $I=V'(\body{\qr})$. By definition of
    minimality, there is no valuation that agrees on the head variables and is
    satisfying for $\qr$ on a strict subset of
    $V'(\body{\qr})$. Therefore, $V'(\head{\qr})$ is in $\qr(I)$ but
    it is not derived on any node and thus  query~$\qr$
    is not parallel-complete under policy~$\distp$.
\end{proof}

\noindent
The characterisation in Lemma~\ref{lem:pc:ucq}, in turn, can be used to prove a $\pitwo$ upper bound.
\begin{lemma}
	\label{lem:paracor-ucqineq-pi2}
	$\paracor(\classUCQineq,\classp)$ is in
        $\pitwo$, for every $\classp\in\Classp$.
\end{lemma}
 \begin{proof}
	It suffices to show that the complement of
	$\paracom(\classUCQineq,\classpnondet^k)$ is in $\sigmatwo$ for arbitrary
	$k \ge 2$. Let $\distp=(n,T)$ be a policy from $\classpnondet^k$. We have to
	consider only instances whose data values can be represented by strings of
	length~$n$ over networks whose nodes can be represented by strings of
	length~$n$.
	
	By Lemma~\ref{lem:pc:ucq}, a query~$\qr$ is not parallel-correct under
	distribution policy~$\distp$ if and only if there exists a minimal
	valuation~$V$ that satisfies $\qr$ on some instance $I$ with $\adom{I}
	\subseteq \dom_n$ such that no node in $\dom_n$ is responsible for all facts
	from $V(\body{\qr})$.
	
	First, the algorithm non-deterministically guesses a 
	valuation~$V$, which can be represented by a string in length polynomial in $\qr$ and $n$.
	Subsequently, it checks for all valuations~$V'$, all nodes~$\node$,
	and all strings $x$ of polynomial length whether $V'$ contradicts minimality of
	$V$ (in which case the algorithm rejects the input) and, by use of
	algorithm~$T$, whether node~$\node$ is not responsible for at least one fact
	from $V(\body{\qr})$ (if so, the algorithm continues, otherwise it rejects).
	All tests can be done in polynomial time.
\end{proof}

\noindent
From \cite{DBLP:conf/pods/AmelootGKNS15}  we know the following result.
\begin{theorem}[\cite{DBLP:conf/pods/AmelootGKNS15}]
	\label{thm:paracor-pi2-complete}
$\paracor(\classCQ,\classpfin)$ is $\pitwo$-complete.
\end{theorem}

Together with Lemma~\ref{lem:paracor-ucqineq-pi2} we get the following result.

\begin{theorem}\label{theo:ucq:compl}
	\label{thm:paracor-ucq-pi2-complete}
$\paracor(\classUCQineq,\classp)$ is $\pitwo$-complete, for every $\classp\in\Classp$.
\end{theorem}

\newcommand{\pred}[1]{\EM{\texttt{#1}}}
\renewcommand{\vec}[1]{\EM{\boldsymbol{#1}}}
\newcommand{\wronglycoloredsubquery}{\EM{\qr_2^{2}}}
\newcommand{\missinglabelsubquery}{\EM{\qr_2^{1}}}
\newcommand{\labelfunction}{\EM{\textit{label}}}

\section{Containment of \CQneg and \UCQneg}
\label{sec:containment}

In this section, we establish the complexity of containment for
\classCQneg and \classUCQneg. We need these results to establish
lower bounds on parallel-correctness and its constituents in the next section. 
Whereas containment for \classCQ has been intensively studied in the literature,
the analogous problems for \classCQneg and \classUCQneg have hardly been addressed
and seem to belong to folklore. In fact, we only found a reference of a complexity result for containment of \classCQneg in \cite{DBLP:journals/tcs/Ullman00}, where a $\pitwo$-algorithm for the problem is given, based on
observations in \cite{DBLP:conf/vldb/LevyS93}, and the existence of a matching lower bound is mentioned. 
However, as we show below, although the problem is indeed in $\pitwo$
for queries defined over a fixed schema (or when the arity of
relations is bounded), it is \conexptime-complete in the general case.

We first show the lower bounds. They actually already hold for Boolean queries.
We show that $\containment(\classBCQneg,\classUBCQneg)$ is
\conexptime-hard by a reduction from the succinct 3-colorability
  problem and afterwards that $\containment(\classBCQneg,\classUBCQneg)$ can be
  reduced to $\containment(\classBCQneg,\classBCQneg)$. Here,
  $\classBCQneg$ and $\classUBCQneg$ denote the class of Boolean
    \CQneg{}s and unions of Boolean \CQneg{}s, respectively. Together this
  establishes that $\containment(\classBCQneg,\classBCQneg)$ and
  therefore also $\containment(\classCQneg,\classCQneg)$ are \conexptime-hard.

\begin{proposition}\label{prop:cqneg-hard}
    $\containment(\classBCQneg,\classUBCQneg)$ is \conexptime-hard.
\end{proposition}
\begin{proof}
  The proof is by a reduction from the succinct 3-colorability
  problem, which asks, whether a graph $G$, which is implicitly given
  by a circuit with binary AND- and OR- and unary NEG-gates, is $3$-colorable. The latter problem is known to be
\NEXPTIME-complete \cite{DBLP:journals/iandc/PapadimitriouY86}.
We say that a circuit $C$, with $2\ell$ Boolean inputs, describes a graph
$G=(N,E)$, when $N = \{0,1\}^\ell$, and there is an edge $(n_1, n_2) \in N^2$ if and only if $C$ outputs true on input $n_1n_2$. 

Let $C$ be an input for the succinct 3-colorability problem with $2\ell$ Boolean
inputs. We construct queries $\qr_1$ and $\qr_2$ such that $\qr_1 \not\subseteq \qr_2$ if and only if the graph described by $C$ is 3-colorable. 

Both queries are over schema $\sch$, which consists of relation names
$\pred{DomainValues}^{(3)}$, $\pred{Bool}^{(1)}$, $\pred{And}^{(3)}$,
$\pred{Or}^{(3)}$, $\pred{Neg}^{(2)}$, and $\pred{Label}^{(\ell+1)}$.
Intuitively, satisfaction of $\qr_1$ will guarantee that there is a
tuple $(a_0,a_1,a_2)$ with three different values in relation
$\pred{DomainValues}$. We will use, for some such tuple,  $a_0,a_1,a_2$ as colors and $a_0,a_1$ as
truth values. We will often assume without loss of generality that
$(a_0,a_1,a_2)=(0,1,2)$. In particular, for such a tuple, $a_0$
is interpreted as false while $a_1$ is interpreted as true.
The unary relation $\pred{Bool}$ will be forced by $\qr_1$ to contain
at least $a_0$ and $a_1$.

Relations $\pred{And}$, $\pred{Or}$, and $\pred{Neg}$ are intended to
represent the respective logical functions. The first two attributes represent input
values, and the last attribute represents the output. Again, $\qr_1$
will guarantee that at least all triples of Boolean values that are
consistent with the semantics of AND, OR, and NEG are present in these relations.
Tuples in relation $\pred{Label}$
represent nodes together with their respective color (one can think of the
representation of a node by $\ell$-ary addresses over a ternary alphabet).

\noindent
We define query $\qr_1$ as follows:
\begin{align*}
    T() \leftarrow & \pred{DomainValues}(w_0, w_1, w_2), \lnot \pred{DomainValues}(w_1,
w_0,
w_2),\\
& \lnot \pred{DomainValues}(w_2, w_1, w_0), \lnot \pred{DomainValues}(w_0,
w_2, w_1), \\
& \pred{Bool}(w_0), \pred{Bool}(w_1),  \pred{Neg}(w_1, w_0), \pred{Neg}(w_0, w_1),\\
& \pred{And}(w_0,w_0, w_0), \pred{And}(w_0,w_1, w_0), \pred{And}(w_1, w_0,
    w_0), \pred{And}(w_1, w_1, w_1), \\
& \pred{Or}(w_0, w_0, w_0), \pred{Or}(w_0, w_1, w_1), \pred{Or}(w_1,w_0,
w_1), \pred{Or}(w_1, w_1, w_1). 
\end{align*}
It is easy to see that $\qr_1$ enforces the conditions mentioned above.

In the following, we denote sequences $x_1,\ldots,x_\ell$ of $\ell$
variables by $\vec x$.

\noindent
We define $\qr_2$ as the union of the queries
$\missinglabelsubquery$ and $\wronglycoloredsubquery$,
where subquery $\missinglabelsubquery$ is defined as:
\begin{align*}
    T() \leftarrow &\pred{Bool}(x_1), \pred{Bool}(x_2), \ldots, \pred{Bool}(x_\ell),
    \pred{DomainValues}(y_r, y_g, y_b), \\
    & \lnot \pred{Label}(\vec x, y_r),
    \lnot \pred{Label}(\vec x, y_g),
    \lnot \pred{Label}(\vec x, y_b).
\end{align*}
Intuitively, $\missinglabelsubquery$ can be satisfied in a database if
for some node, represented by $\vec x$, there is no color.  

Subquery $\wronglycoloredsubquery$ deals with the correctness of a
coloring and uses a set $\textsc{circuit}$ of atoms that is intended
to check whether for two nodes $u$ and $v$, represented by $\vec y$
and $\vec z$, respectively, there is an edge between $u$ and $v$.

To this end, $\textsc{circuit}$ uses the variables
$y_1,\ldots,y_\ell$, $z_1,\ldots,z_\ell$, representing the input and, at
the same time, the $2\ell$ input gates of $C$, and an additional variable~$u_i$,
for each gate of $C$, with the exception of the output
gate. The output gate is represented by variable $w_1$.
For each AND-gate represented by variable
$v_1$ with incoming edges from gates
represented by variables $u_1$ and $u_2$, $\textsc{circuit}$ contains
an atom $\pred{And}(u_1,u_2,v_1)$. Likewise for OR- and NEG-gates. 

Subquery $\wronglycoloredsubquery$ is defined as:
\begin{align*}
T() \leftarrow &  \pred{DomainValues}(w_0, w_1, w_2), 
\textsc{circuit},
\pred{Label}(\vec{y}, u), \pred{Label}(\vec{z}, u).
\end{align*}

Intuitively, $\wronglycoloredsubquery$ returns true when two nodes, witnessed to be adjacent by the
circuit, have the same color. 

We show in the appendix that $C$ is 3-colorable if 
and only if  $\qr_1 \not\subseteq \qr_2$.
\end{proof}

Next, we provide the above mentioned reduction.

\begin{proposition} \label{prop:red:simplify:containment}
$\containment(\classBCQneg,\classUBCQneg) \polred
\containment(\classBCQneg,\classBCQneg)$ 
\end{proposition}
\begin{proof}
    Let $\qr_1$ be in \classBCQneg and $\qr_2=\bigcup_{i=1}^m \qr_2^i$ be in \classUBCQneg over some database schema $\sch$. Recall our assumption, that each disjunct is defined over a disjoint set of variables.
    Next, we construct \CQs $\qr_1'$ and $\qr_2'$ such that
    $\qr_1' \subseteq \qr_2'$ if and only if, $\qr_1 \subseteq \qr_2$. 
    
    We explain the intuition behind the reduction by means of an example. To
    this end, let $\qr_1$ be $H() \gets A(x,y)$ and let $\qr_2$ be
    the $\qr^1_2 \cup \qr^2_2$, where $\qr_2^1$ is $H()
    \gets A(u_1,v_1),B(u_1,v_1)$ and $\qr_2^2$ is $H() \gets
    A(u_2,v_2),\neg B(u_2,v_2)$, both formulated over the schema $\sch =
    \{A^{(2)},B^{(2)}\}$. The query $\qr'_2$ takes the following form: 
     \begin{align*}    
     H() \gets
     \pred{Active}(x_0,x_1;\ell_1,\ell_2),
     \underbrace{\alpha(\ell_1, \qr^1_2)}_{{\qr'_{2,1}}}, \underbrace{\alpha(\ell_2, \qr^2_2)}_{{\qr'_{2,2}}},
    \end{align*}
    where $\alpha(w,\qr)$ denotes the modification of the body of $\qr$ by
replacing every atom $\pred{R}(\vec{x})$ by $\pred{R}'(w,
\vec{x})$.
     Both queries are defined over the schema $\sch'
     =\{A'^{(3)},B'^{(3)},{\pred{Active}}^{(4)}\}.$ Notice that $\qr'_2$ contains
     a concatenation of the disjuncts of $\qr_2$. In addition, relations
     $A$ and $B$ are extended with a new first column with the purpose of
     labelling tuples. This labelling allows to encode two (or even more)
     instances over $\sch$ by one instance over $\sch'$. Specifically, $\body{\qr'_1}$ (not shown)
     is constructed in such a way that when there is a satisfying valuation for
    $\qr'_1$ there are two different data values, say 0 and 1. So, an instance $I$ over
    $\sch$ can be encoded as $I^0 = \{A'(0,a,b)\mid A(a,b)\in I\} \cup \{B'(0,a,b)\mid
    B(a,b)\in I\}$ or as $I^1 = \{A'(1,a,b)\mid A(a,b)\in I\} \cup \{B'(1,a,b)\mid B(a,b)\in
    I\}$. In addition, when there is a satisfying valuation for $\qr'_1$, there is an
    instance $I_{2}$ on which 
    every disjunct of $\qr_2$ is true, and there
    is an instance $I_1$ on which $\qr_1$ is true. So, both $\qr'_{2,1}$ and $\qr'_{2,2}$
    evaluate to true on $I_{2}^0$ when $\ell_1$ and $\ell_2$ are interpreted by label
    $0$. However, for $\qr_1$ to be contained in $\qr_2$, we need that at least one of the
    disjuncts $Q'_{2,1}$ or $Q'_{2,2}$ evaluates to true over $I_1^1$,
    that is, when its labelling variable is interpreted as 1. Atom
    $\pred{Active}(x_0,x_1;\ell_1,\ell_2)$
    will ensure that $x_0$
    and $x_1$ correspond with the values 0 and 1, and that at least one of the labelling variables $\ell_1$ or $\ell_2$ is equal to $1$. In other words, $\pred{Active}$ chooses which disjunct to activate over $I_1$. So, at least one disjunct of $\qr_2$ evaluates to true on the instance $I_1$ on which $\qr_1$ is satisfied.

        We explain the reduction in more detail in the appendix.
    \end{proof}

Combining Propositions~\ref{prop:cqneg-hard} and 
\ref{prop:red:simplify:containment} we get the following corollary:
\begin{corollary}\label{coro:cqneg-hard}
    $\containment(\classCQneg,\classCQneg)$ is \conexptime-hard.
\end{corollary}

\noindent
The corresponding upper bounds hold also in the presence of
inequalities and are shown by small model (i.e.,
counter-example) properties. To this end, we make use of a restricted
monotonicity property of \UCQnegineqs which was already observed in Proposition 2.4 of \cite{AfratiCY95}. For an instance $I$ and a
set $D$ of data values we denote by $\restrict{I}{D}$ the restriction
of $I$ to facts that only use values from $D$.

\begin{lemma}[\cite{AfratiCY95}]\label{lem:restrict}
  For $\qr\in\classUCQnegineq$, $I$ an instance
  with a compatible
  schema, and $D$ a set of data values, it holds that $\qr(\restrict{I}{D})\subseteq \qr(I)$.
\end{lemma}

\begin{proof}
  Let $\fc\in \qr(\restrict{I}{D})$ via a valuation $V$ for a disjunct
  $\qr_i$ of \qr. Thus, $V(\pos{\qr_i}) \subseteq\restrict{I}{D}\subseteq I$. 
By definition, every variable $x$ of $\qr_i$ occurs in a
  positive atom and therefore $V(x)\in D$. Thus, $V(\negbody{\qr_i})
  \cap I = V(\negbody{\qr_i}) \cap \restrict{I}{D} =
\emptyset$ and $\fc\in \qr(I)$ as claimed.
\end{proof}

Now we can establish the following small model property
for testing containment.

  \begin{lemma}
Let $\qr_1,\qr_2 \in \classUCQnegineq$. 
If there is an instance $I$, where $\qr_1(I) \not\subseteq \qr_2(I)$, 
then there is also an instance $J\subseteq I$, where $\qr_1(J) \not\subseteq \qr_2(J)$, and $\ssize{\adom{J}} \le \varmax{\qr_1}$.
\label{lem:reduction_helper}%
\end{lemma}
\begin{proof}
    Let $I$ be as in the lemma and let $\fc$ be a fact with $\fc \in \qr_1(I)$ and
    $\fc \not \in \qr_2(I)$.  Let $V$ be a valuation that derives \fc via
    some disjunct
$\qr^i_1$ of $\qr_1$.

Let $D\mydef \adom{V(\pos{\qr_1^i})}$ and  $J\mydef \restrict{I}{D}$ the set of all facts in $I$ using only values from
$\adom{V(\pos{\qr_i})}$.
By definition, $\ssize{\adom{J}} \le \varmax{\qr_1}$. Clearly, $V$ is still a
satisfying valuation for $\qr^i_1$ over $J$. However, by
Lemma~\ref{lem:restrict}, $\fc\not\in\qr_2(J)=\qr_2(\restrict{I}{D})$.  
\end{proof}

The upper bounds follow easily from Lemma~\ref{lem:reduction_helper}.
\begin{proposition}\label{prop:UCQneg:upper} 
The following upper bounds hold:
\begin{enumerate}
\item $\containment(\classUCQnegineq,\classUCQnegineq)$ is in \conexptime.
\item For every $k$, containment of $\classUCQnegineq$-queries over
  schemas with arity bound $k$ is in $\pitwo$.
\end{enumerate}
\end{proposition}
\begin{proof}
	In both cases, we consider the complement of
	$\containment(\classUCQnegineq,\classUCQnegineq)$.
    Let $m\mydef \varmax{\qr_1}$. 
  \begin{enumerate}
  \item 
    A \nexptime algorithm, on input $\qr_1,\qr_2$, can simply guess an
    instance $J$ with a domain of at most $m$ elements and a fact \fc,
    and verifies that $\fc\in\qr_1(J)$ but   $\fc\not\in\qr_2(J)$. For
    the latter tests, it can simply cycle, in exponential time, to all
    valuations over $J$ for $\qr_1$ and $\qr_2$. 
  \item For a fixed arity bound, the minimal counter-example $J$ is of
    size at most $m^k$. It can thus be guessed in polynomial
    time. That  $\fc\in\qr_1(J)$ can be verified
    non-deterministically. That $\fc\not\in\qr_2(J)$ can
    be verified by a universal computation in polynomial time.
  \end{enumerate}
\end{proof}
A claim of a $\pitwo$ upper bound for containment of CQs with negation
can be found in \cite{DBLP:journals/tcs/Ullman00}. It was not made
clear there, that this claim assumes bounded arity of the schema. That
the containment problem is $\pitwo$-complete for schemas of bounded
arity has been explicitly shown  in \cite{MugnierST12}. Clearly,
Proposition~\ref{prop:UCQneg:upper}.2  follows directly and
\ref{prop:UCQneg:upper}.1 is only
a variation of it.
From Proposition~\ref{prop:UCQneg:upper} and
Corollary~\ref{coro:cqneg-hard} the main result of this section
immediately follows. 

\begin{theorem}
	\label{thm:containment-cqneg-conexptime}
	$\containment(\classBCQneg,\classBCQneg)$ and
        $\containment(\classUCQnegineq,\classUCQnegineq)$ are
        \conexptime-complete.
\end{theorem}
Of course, the theorem also holds for all classes $\cal C$ of queries
with $\classBCQneg\subseteq\cal C\subseteq \classUCQnegineq$.

\section{Parallel-correctness: unions of conjunctive queries with negation}
\label{sec:negation}

As mentioned in Section~\ref{sec:ucq}, for conjunctive queries without
negation parallel-soundness always holds and thus parallel-correctness
and parallel-completeness coincide, thanks to monotonicity. For queries with
negation the situation is different. Distributed evaluation can be complete but not sound, or vice versa. For this reason, we
have to distinguish all three problems separately: correctness, soundness, and
completeness. However, the complexity is the same in all three cases.

Our results show a second, more crucial difference. Whereas parallel
completeness for CQs without negation could be characterised in terms
of valuations, that is, objects of polynomial size, our
algorithms for CQs with negation involve counter-examples of
exponential size (if the arity of  schemas is not bounded) and the $\coNEXPTIME$ lower bound results indicate that this is unavoidable.
We illustrate the observation that counter-examples might need an
exponential number of tuples by the following example.

\begin{example}
	Let $\qr$ be the following conjunctive query with negation:
	\begin{displaymath}
		\begin{array}{lcl}
			H()
			& \gets &
			\Dom(w_0,w_0), \Dom(w_1,w_1),
			\Dom(x_1,x_1),\dots,\Dom(x_n,x_n), \\
			& & \lnot\Dom(w_0,w_1), \lnot \Rel(x_1,\dots,x_n).
		\end{array}
	\end{displaymath} 
	Let $\distp$ be the policy defined over universe $U=\{0,1\}$ and
	two-node network $\{\cnodeA,\cnodeB\}$, which distributes all facts except
	$\Rel(0,\dots,0)$ to node~$\cnodeA$ and only fact $\Rel(0,\dots,0)$ to
	node~$\cnodeB$.

	Query~$\qr$ is not parallel-sound under policy~$\distp$, but
        the smallest counter-example $I$ is of
	exponential size as we argue next. 
    Indeed, let $I \mydef \{\Dom(0,0), \Dom(1,1)\} \cup \{\Rel(a_1,\dots,a_n) \mid
	(a_1,\dots,a_n) \in \{0,1\}^n\}$. Furthermore, let valuation $V$ map variables
	$w_1$ and $w_0$ to $1$ and $0$, respectively, and map $x_i$ to 0, for every
	$i \in \{1,\dots,n\}$. Then, valuation~$V$ satisfies $\qr$ on instance
	$\dist{I}(\cnodeA) = I \setminus \{\Rel(0,\dots,0)\}$ because neither
	$\Dom(0,1)$ nor $\Rel(0,\dots,0)$ is contained in the local instance.
	Furthermore, there is no satisfying valuation~$W$ for $\qr$ on the global
	instance~$I$ because $W$ would have to map each $x_i$ to either $0$ or $1$
	implying that 
    $W\big(\Rel(x_1,\dots,x_n)\big)\in I$.

	However, there is no smaller instance: let $I^\ast$ be some instance
	over universe $U$ that has a locally satisfying valuation~$V$. The
	combination of atoms $\Dom(w_0,w_0),\Dom(w_1,w_1)$, and $\lnot\Dom(w_0,w_1)$ in
	query~$\qr$ then implies existence of both facts $\Dom(0,0)$ and $\Dom(1,1)$
	because variables $w_0$ and $w_1$ cannot be mapped onto the same data value. 

	Assume
	that fact $\Rel(a_1,\dots,a_n)$, for some $(a_1,\dots,a_n) \in \{0,1\}^n$ is missing from $I^\ast$. Then the valuation~$W$ that maps $w_0 \mapsto 0, w_1
	\mapsto 1$ and $x_i \mapsto a_i$, for every $i \in \{1,\dots,n\}$, satisfies
	$\qr$ also globally, on instance~$I^\ast$, and can therefore be no example
	against parallel-soundness, which contradicts our choice of $I^\ast$. Thus,
	$\Rel(a_1,\dots,a_n) \in I^\ast$, for every $(a_1,\dots,a_n) \in \{0,1\}^n$. We
	therefore have $I \subseteq I^\ast$ and, in particular, instance~$I^\ast$
	contains at least as many facts as instance~$I$. \hfill $\Box$
\end{example}

The results of this section are summarised in the following theorem:

\begin{theorem}\label{theo:paracomp-main}
For every class
                  $\classp\in\{\classprule\}\cup\Classpnondet$ of
                  distribution policies, the following problems  are
                 $\coNEXPTIME$-complete.
                 \begin{itemize}
                 \item $\parasound(\classUCQneg,\classp)$
                 \item
                   $\paracom(\classUCQneg,\classp)$
                 \item                    $\paracor(\classUCQneg,\classp)$
                 \end{itemize}
\end{theorem}
Theorem~\ref{theo:paracomp-main} follow from
Propositions~\ref{prop:paracomp-upper} and \ref{prop:paracomp-lower}
below. It also holds for
$\classUCQnegineq$. It is easy to show that, when restricted to
schemas with some fixed (but sufficiently large, for hardness) arity bound, all these problems are $\pitwo$-complete.

\subsection{Upper bounds}

In this section, we show the upper bounds of
Theorem~\ref{theo:paracomp-main}, summarised in the following
proposition.
\begin{proposition}\label{prop:paracomp-upper}
  
	 $\parasound(\classUCQnegineq,\classp)$,
		 $\paracom(\classUCQnegineq,\classp)$, and
		 $\paracor(\classUCQnegineq,\classp)$ are
                 in $\coNEXPTIME$, for every class
                  $\classp\in\Classp$ of
                  distribution policies. If the arity of schemas is
                  bounded by some fixed number, these problems are in \pitwo.
\end{proposition}

\begin{proof}
  As already indicated above, the proof relies on a bound on the size of
a smallest counter-example. More specifically, we first show the
following claim.
\begin{claim}\label{claim:smallmodel}
Let $\qr\in\classUCQnegineq$ and let \distp be an arbitrary distribution
policy. Then the following statements hold:
  \begin{enumerate}
  \item If \qr is not parallel-complete under \distp, then there
    is an instance $J$ over a domain with at most $\varmax{\qr}$
    elements such that \qr is not parallel-complete on $J$ under \distp.
  \item If \qr is not parallel-sound under \distp, then there
    is an instance $J$ over a domain with at most $\varmax{\qr}$
    elements such that \qr is not parallel-sound on $J$ under \distp.
  \end{enumerate}
\end{claim}
Towards (1) let us assume that \qr is not parallel-complete on some
instance $I$ under \distp. Let $V$ be a valuation of a disjunct
$\qr_i$ of $\qr$ that derives a fact
\fc globally that is not derived on any node of the network. Let
$D\mydef\adom{V(\pos{\qr_i})}$ and
$J\mydef\restrict{I}{D}$. Clearly, $|D|\le\varmax{\qr}$ and 
$V$ still derives \fc globally on
instance $J$ via
$\qr_i$. On the other hand, for every node \node,
$\qr\big(\dist{J}(\node)\big)
= \qr\big(\restrict{\dist{I}(\node)}{D}\big) \subseteq
\qr\big(\dist{I}(\node)\big)$, thanks to Lemma~\ref{lem:restrict}.
Therefore \fc is not derived on \node, and thus $J$ witnesses the lack of
parallel-completeness of \qr under \distp.

The proof of (2) is completely analogous. Given a counter-example $I$
and a valuation $V$ that
derives a fact \fc on some node \node via $\qr_i$, for which $\fc$ is not
derived globally, we define $D\mydef\restrict{I}{\adom{V(\pos{\qr_i})}}$ and
show that $J\mydef\restrict{I}{D}$ is the desired counter-example.

In the appendix, we describe an algorithm that tests the complement of parallel
completeness non-deterministically.
\end{proof}

\subsection{Lower bounds}
\label{sec:reduction:containment}
The lower bounds stated in Theorem~\ref{theo:paracomp-main} follow from a
polynomial time reduction from problem
$\containment(\classBCQneg,\classBCQneg)$, for which we showed  $\coNEXPTIME$-hardness in Section~\ref{sec:containment}.

\begin{proposition}
	\label{prop:paracomp-lower}
	 $\paracom(\classCQneg,\classpenum)$,  $\parasound(\classCQneg,\classpenum)$,
         and $\paracor(\classCQneg,\classpenum)$ are
         $\coNEXPTIME$-hard.
\end{proposition}

\begin{proof}
Interestingly, all three results are shown by the \emph{same} reduction from
decision problem $\containment(\classBCQneg,\classBCQneg)$. 

	The basic idea for this reduction is very simple: it combines both
	queries $\qr_1,\qr_2 \in \classBCQneg$ of the given  containment instance into a
	single query $\qr \in \classBCQneg$ and infers an appropriate
        distribution policy $\distp$.
	To emulate separate derivation for both queries in the combined query, 
    an activation mechanism is used that resembles the proof of
    Proposition~\ref{prop:red:simplify:containment}.
    In this fashion, the two queries can be evaluated over different subsets of the
    considered instance by annotating both the facts in the instance as well as the atoms
    of the query.
   
We next describe the reduction in detail. Let thus
$\qr_1,\qr_2\in\classBCQneg$ be queries over some schema $\sch$ and let
$m\mydef \max \big\{\varmax{\qr_1},\varmax{\qr_2}\big\}$. Without loss of generality, we assume  the variable
	sets of $\qr_1$ and $\qr_2$ to be disjoint. We will also
        assume in the following that both $\qr_1$ and $\qr_2$ are 
        satisfiable. This is the case (for $\qr_1$) if and only if
        $\bodypos{\qr_1}\cap\bodyneg{\qr_1}=\emptyset$ and can
        therefore be easily tested in polynomial time. If one of the
        test fails, some appropriate constant instance of
        $\paracom(\classCQneg,\classpenum)$ or one of the other problem
        variants, respectively, can be computed.

	We define a (Boolean) query $\qr \in \classBCQneg$ and a policy $\distp \in
	\classpenum$ over domain $\{1,\dots,m\}$ that  can be computed
        from $\qr_1$ and $\qr_2$ in polynomial time. The schema for \qr is $\sch' \mydef
	\{R'^{(k+1)} \mid R^{(k)} \in \sch\}$. That is, each relation
        name $R$ of \sch occurs as $R'$ in $\sch'$ with an arity incremented
        by one. Additionally, $\qr$ uses relation names $\Type$, $\StartA$,
        $\StartB$, and $\Stop$, which we assume not to occur in schema~$\sch$. Besides the variables of
        $\qr_1$ and $\qr_2$, query~$\qr$ uses variables $\ell_1,\ell_2, t$. 

We use the function $\alpha$, defined in the proof of Proposition~\ref{prop:red:simplify:containment}, which adds its first parameter as
first component to every tuple in its second parameter and translates
relation names $R$ into $R'$. In
Proposition~\ref{prop:red:simplify:containment}, the first parameter
was always a variable and the second a set of atoms, but we use
$\alpha$ also for a data value as first and a set of facts as second
parameter in the obvious way. We write $\alphainv{a}$ for the function
mapping sets of facts over $\sch'$ to sets of facts over $\sch$,
by selecting, from a set of facts, all facts with first parameter $a$,
deleting this parameter and replacing each name $R'$ by
$R$. Finally, $\pi_a(I)\mydef\alpha\big(a,\alphainv{a}(I)\big)$ is the
restriction of $I$ to all facts with $a$ in their first component.

\medskip
\noindent
	The combined query~$\qr$ has $\head{\qr} \mydef H()$ and body
	\begin{displaymath}
		\begin{array}{lcl}
			\body{\qr} &
			\mydef &
			\alpha(\ell_1,\body{\qr_1})
			\cup
				\alpha(\ell_2,\body{\qr_2}) \\
			& \cup &
			\underbrace{\{\Type(t), \StartA(\ell_1), \StartB(\ell_1),
			\StartB(\ell_2)\}}_{\atomset}
			\cup
			\underbrace{\{\lnot\Stop(\ell_1), \lnot\Stop(\ell_2)\}}_{\atomset^\lnot}.
		\end{array}
	\end{displaymath}

	\noindent
	Policy~$\distp$ is defined over universe $U \mydef
        \{1,\dots,m\}$, schema $\sch' \cup \{\Type,\StartA,\StartB,\Stop\}$ and network $\nw \mydef
	\{\nodeComp_1,\dots,\nodeComp_m,\nodeSound_1,\dots,\nodeSound_m,\nodeRest\}$.
	Facts are distributed as follows:

        \begin{itemize}
        \item Every node~$\nodeComp_i$ is responsible for the facts
          $\Type(1),\StartA(i),\StartB(i),\Stop(i)$, and all facts
          from $\facts{\sch',U}$.
        \item Every node~$\nodeSound_i$ is
          responsible for the facts $\Type(2),\StartA(i),\Stop(i)$,
          all $\StartB$-facts, and all facts  from $\facts{\sch',U}$.
        \item Finally, node~$\nodeRest$ is responsible for facts
          $\Type(3),\dots,\Type(m)$, and \emph{all} facts over other
          relation names.
        \end{itemize}
It is easy to see that $\distp$ can be
        expressed by a polynomial number of rules and that \qr and
        \distp can be computed in polynomial time.
        In the appendix, we show  that the described function is indeed the desired reduction.
\end{proof}

\section{Full conjunctive queries}
\label{sec:full}

In this section, we focus attention on full conjunctive queries, in an attempt to lower the complexity of testing parallel-correctness.
Requiring queries to be full is a very natural restriction which is known to have practical benefits.
For example, the Hypercube algorithm, which describes an optimal way to compute \CQ{}s in a setting very similar to
ours, completely ignores projections when shuffling data, and only applies them when computing the query locally.
The latter is possible because correctness for the full-variant of a query is in a sense more strict than correctness for the query itself.

Formally, a (union of) conjunctive queries is called \emph{full} if all variables
of the body also occur in the head.
We denote by $\classFCQnegineq$ and $\classFUCQnegineq$ the class of full $\CQnegineq$
and full $\UCQnegineq$ queries, respectively, and likewise for other fragments.  

The presentation is similar to that of Section~\ref{sec:containment} and \ref{sec:negation}. First, we establish
the complexity of query containment. Then, we show that containment reduces
to parallel-correctness (and variants). Finally, we obtain matching upper bounds.

The following theorem shows that unlike for general conjunctive queries the complexity of deciding containment for \FCQneg and
\UFCQneg do not coincide. 
\begin{theorem}
    \begin{enumerate}
        \item $\text{\sc Containment}(\classFCQneg, \classFCQneg)$ is in $\Poly$;
        \item 
    $\text{\sc Containment}(\classFCQneg, \classFUCQneg)$ is \coNP-complete; and
    \item 
    $\text{\sc Containment}(\classUFCQneg, \classFUCQneg)$ is \coNP-complete.
    \end{enumerate}\label{theo:full_containment_new}
All these results also hold for queries with inequalities.
\end{theorem}

As one can reduce from $\text{\sc Containment}(\classFCQneg, \classFUCQneg)$ to parallel-soundness, completeness,
and correctness, we obtain the following hardness results:

\begin{proposition}
    $\parasound(\classUFCQneg,
            \classp)$, $\paracom(\classUFCQneg,
            \classp)$, and $\paracor(\classUFCQneg,\classp)$ are \coNP-hard, for
            every $\classp \in \{\classpenum\} \cup \Classpnondet$.
            \label{the:full_parallelXreductions:new}
\end{proposition}

The following theorem determines the complexity for the upper bounds:

\begin{theorem}
    The following problems are $\coNP$-complete:
\begin{enumerate}
\item $\parasound(\classFUCQneg,\classpenum)$;
\item $\paracom(\classFUCQneg,\classpenum)$;
\item $\paracor(\classFUCQneg,\classpenum)$.
\end{enumerate}\label{the:parallelX_conpcomplete_full}
The result also holds for queries with inequalities.
\end{theorem}

\section{Discussion}
\label{sec:discussion}

In this paper, we continued the study of parallel-correctness initiated by Ameloot et al.~\cite{DBLP:conf/pods/AmelootGKNS15} as a framework for reasoning about one-round evaluation algorithms for conjunctive queries under arbitrary distribution policies. Specifically, we considered the case with union and negation. While parallel-correctness for unions of conjunctive queries can be tested by examining properties of single valuations, just like in the union-free case, the latter no longer holds true when negation is present. Consequently, we obtained that deciding parallel-correctness for unions of conjunctive queries remains in \pitwo, while the analog problem in the presence of negation is hard for \conexptime. Since conjunctive queries with negation are no longer monotone, we considered the related problems of parallel-completeness and parallel-soundness as well and obtained the same bounds. Interestingly, when negation is present, containment of conjunctive queries can be reduced to parallel-correctness (and its variants) allowing the transfer of lower bounds. We prove that containment for conjunctive queries with negation is hard for \conexptime, which, to the best of our knowledge, is a novel result. In an attempt to lower complexity, we show that parallel-correctness for unions of full conjunctive queries with negation is \coNP-complete. 

There are quite a number of directions towards future work. 
While parallel-correctness for first-order logic is undecidable, it would be interesting to determine the exact frontier for decidability. 
As the considered problem is a static analysis problem that relates to the size of the queries and not to the size of the instances (at least in the setting of $\classpenum$), exponential lower bounds do not necessarily 
exclude practical application. It could still be interesting to identify settings that would make parallel-correctness tractable. Possibly independent of tractability considerations, such settings could incorporate bag semantics, integrity constraints, or specific classes (and representations) of distribution policies. We also plan to consider evaluation algorithms that use knowledge about the distribution policy to compute better query results, locally. 
Another direction for future work is to investigate transferability of parallel-correctness for conjunctive queries as defined in \cite{DBLP:conf/pods/AmelootGKNS15} in the presence of union and negation.

\bibliographystyle{plain}
\bibliography{references}

\section*{Appendix}

\section{Proofs for Section~\ref{sec:containment}: Containment of \CQneg and \UCQneg}

\begin{proof}[Proof of Proposition~\ref{prop:cqneg-hard} (continued)]

It remains to show that the graph represented by circuit~$C$ is 3-colorable if 
and only if $\qr_1 \not\subseteq \qr_2$.

\IfDirection
Suppose $\qr_1 \not \subseteq \qr_2$. Thus, for some instance $I$, $T() \in \qr_1(I)$ while $T() \not \in
\qr_2(I)$. The former implies that $I$ contains a fact $\pred{DomainValues}(a_0,a_1,a_2)$, where $a_0,a_1,$ and $a_2$
are distinct values from $\dom$. Without loss of generality, we assume
$a_0=0, a_1=1, a_2=2$. The instance $I$ further contains the facts $\pred{Bool}(0)$
and $\pred{Bool}(1)$, and all \enquote{logical facts} induced by $\qr_1$. 

For every vector $\vec
n\in\{0,1\}^\ell$, there must be some $c\in\{0,1,2\}$ such
that $\pred{Label}(\vec n;c)\in I$, since otherwise $T()\in\qr^1_2$. Let
$\labelfunction:\{0,1\}^\ell\to\{0,1,2\}$ be chosen
such that,  for
every $\vec{n}\in\{0,1\}^\ell$, $\labelfunction(\vec{n}) = c$, for some $c \in \{0,1,2\}$, for which $\pred{Label}(\vec{n};c) \in I$.

We claim that  $\labelfunction$ is a valid coloring of the graph
represented by $C$. Towards a contradiction let us assume that there
are two nodes $\vec n$ and $\vec n'$ that are connected by an edge and
for which $\labelfunction(\vec
n)=\labelfunction(\vec n')=c$, for some $c$. Then $\qr^2_2$ could be satisfied over $I$
by choosing a valuation that corresponds to a computation of $C$ that
witnesses that there is an edge between $\vec n$ and $\vec n'$ and
mapping $u$ to $c$, the desired contradiction.

\OnlyIfDirection
Let, for some $\ell$, $C$ be a circuit with input length $2\ell$,
that describes a $3$-colorable graph $G$. Let
$\labelfunction:\{0,1\}^\ell\to\{0,1,2\}$ be a valid coloring for
$G$.  Let $I$ be the database with the following facts:
\begin{multline*}
  \{ \pred{Label}\big(\vec{n};\labelfunction(\vec{n})\big) \mid \vec{n} \in
\{0,1\}^\ell\} \cup \{ \pred{DomainValues}(0,1,2), \pred{Bool}(0),
\pred{Bool}(1) \}\cup\\
\{\pred{Neg}(1, 0), \pred{Neg}(0, 1), \pred{And}(0,0, 0), \pred{And}(0,1, 0), \pred{And}(1, 0,
    0), \pred{And}(1, 1, 1)\}\cup\\
\{ \pred{Or}(0, 0, 0), \pred{Or}(0, 1, 1), \pred{Or}(1,0,
1), \pred{Or}(1, 1, 1)\}.
\end{multline*}

Obviously, $T() \in \qr_1(I)$ and $T()\not\in\qr^1_2(I)$. However, since
$I$ only contains the \enquote{correct} logical facts, to
satisfy $\qr^2_2$ it would be necessary to find two nodes with the
same label whose adjacency is witnessed by the canonical valuation
corresponding to the semantics of $C$, which does not exist. Thus, $\qr_1 \not \subseteq \qr_2$. 
\end{proof}

\begin{proof}[Proof of Proposition \ref{prop:red:simplify:containment} (continued)]

    We now explain the reduction in more detail. We assume that
    $\CQneg$ $\qr_2^i$ is satisfiable, for every $i \in
    \{1,\dots,m\}$.\footnote{Notice that a \CQneg $\qr$ is satisfiable if and
    only if $\pos{\qr}\cap \negbody{\qr} = \emptyset$, which can easily be verified in polynomial time.} Set $\sch'\mydef \{\pred{R}'^{(k+1)} \mid \pred{R}^{(k)} \in \sch\} \cup
    \{\pred{Active}^{(m+2)}\}$.        
    As explained above, the relation $\pred{Active}$ serves as an index for the
    disjuncts in $\qr_2$. Whereas the first two positions encode the bits zero
    (0) and one (1), an atom of the form
    $\pred{Active}(0,1;0,\ldots,0,1,0,\ldots,0)$ with a 1 occurring on
    position $i+2$ is meant to indicate disjunct $i$. 
Recall that $\alpha(w,\qr)$ denotes the modification of the body of
$\qr$ by replacing every atom $\pred{R}(\vec{x})$ by $\pred{R}'(w, \vec{x})$,
while retaining existent negation symbols. We further define\footnote{This is
the same definition as in Section~\ref{sec:reduction:containment} and is
stated here for convenience only.} mapping $\alphainv{a}$, for $a \in \dom$, to
map sets of facts over $\sch'$ to sets of facts over $\sch$, by selecting all
facts with first parameter~$a$, deleting this parameter and replacing each
relation name $R'$ by $R$.

\noindent
Now, define $\qr_1'$ as:
    \begin{align*}
H() \leftarrow 
& \pred{Active}(w_0, w_1; w_1, w_0, \ldots, w_0),
\pred{Active}(w_0, w_1; w_0, w_1, \ldots, w_0), 
  \dots,\\
 & \pred{Active}(w_0, w_1; w_0, w_0, \ldots, w_1, w_0),
  \pred{Active}(w_0, w_1; w_0, \ldots, w_0,w_1), \\
&  \lnot\pred{Active}(w_1, w_0; w_1, w_0, \ldots, w_0), \\
& \alpha(w_1, {\qr_1}),
\alpha(w_0, {\qr_2^1}), \ldots, \alpha(w_0, {\qr_2^m});
\end{align*}
and $\qr_2'$ as:
\begin{align*}
H() \leftarrow \pred{Active}(w_0, w_1; z_1, \ldots, z_m),
\alpha(z_1, {\qr_2^1}), \ldots, \alpha(z_m, {\qr_2^m}),
\end{align*}
where $w_0, w_1, z_1, \ldots, z_m$ are fresh and distinct variables.

\medskip
\noindent
In the remainder we argue that
$\qr'_1 \subseteq \qr'_2$ if and only if $\qr_1 \subseteq
\qr_2$:

\IfDirection
Suppose $\qr_1 \subseteq \qr_2$. 
Let $I'$ be an arbitrary instance over $\sch'$ for which $H() \in \qr_1'(I)$. 
Let $V_1'$ be a satisfying  valuation for $\qr_1'$ over $I'$. 
Since
$V'_1$ has to satisfy both literals $\pred{Active}(w_0, w_1; w_1, w_0, \ldots,
w_0)$ and $\lnot\pred{Active}(w_1, w_0; w_1, w_0, \ldots, w_0)$, it
follows that $V'_1(w_1)\not=V'_1(w_0)$. For simplicity, we assume that
$V'_1(w_1)=1$ and $V'_1(w_0)=0$.

 Thus, $I'$ contains the facts $\pred{Active}(0,1;1,0,\ldots, 0)$, \ldots, 
$\pred{Active}(0,1;0,\ldots, 0, 1)$, but not $\pred{Active}(1,0;1,0,\ldots, 0)$.
Furthermore, for every $i$,  $V_1'\big(\alpha(w_0, {\qr_2^i})\big)$ consists of
facts from $I'$ labelled with $0$, and for $V_1'\big(\alpha(w_1, {\qr_1})\big)$
consists of facts from $I'$ labelled with $1$. 

Now, let  $I$ be the set $\alphainv{1}(I') = \{ R(\tup{t}) \mid R'(1, \tup{t})
\in I'\}$ consisting of all facts in
$I'$ labelled with $1$. By construction and since $I'$ contains all
facts from $V_1'\big(\alpha(w_1, {\qr_1})\big)$, $I$ contains all facts from
$V_1'(\qr_1)$. Note that $V_1'$ is a valuation for $\qr_1$ as well.
In fact, $V'_1$ is a satisfying valuation for $\qr_1$
over $I$. Consequently, thanks to $\qr_1 \subseteq \qr_2$, there must be a satisfying valuation $V_2$ for
some disjunct $\qr_2^i$ of $\qr_2$ over $I$. 

We define a satisfying valuation $V_2'$ for $\qr_2'$ over $I'$
witnessing $\qr'_1 \subseteq \qr'_2$ as follows: 
\[
V_2'(x) =
\begin{cases}
  0 & \text{if $x=w_0$},\\
  1 & \text{if $x=w_1$},\\
  1 & \text{if $x=z_i$},\\
  0 & \text{if $x=z_j$ and $j\not=i$},\\
  V_2(x) & \text{if $x$ occurs in $\qr^i_2$},\\
  V'_1(x) & \text{if $x$ occurs in $\qr^j_2$ and $j\not=i$}.\\
\end{cases}
\]
It is easy to verify that $V'_2$ satisfies $\qr'_2$: the first atom
and all conjuncts $\alpha(z_j, {\qr_2^j})$ with $j\not=i$ become true since the respective facts were guaranteed by
$\qr'_1$. Finally, $\alpha(z_i, {\qr_2^i})$ becomes true since $V_2$
satisfies $\qr_2^i$.

\OnlyIfDirection
The proof is by contraposition, that is, we show that $\qr_1
\not\subseteq \qr_2$ implies $\qr'_1 \not\subseteq \qr'_2$.
Therefore, let $I$ be an instance, and $V_1$ a valuation such that
$V_1$ satisfies $\qr_1$ over $I$, but no $\qr^i_2$ has a satisfying
valuation over $I$. Since every query $\qr^i_2$ is satisfiable, there is,
for every $i$, a satisfying valuation $V^i$. In fact, these valuations
can be chosen with pairwise disjoint range. Now, we define $I'$ as the
following set of facts:
\begin{displaymath}
	\begin{array}{lcl}
		I'
  		& = &
   		\alpha(1,I) \cup 
   		\alpha\big(0,V^1(\pos{\qr^1_2})\big) \cup 
   		\dots \cup
   		\alpha\big(0,V^m(\pos{\qr^m_2})\big) \\
   		& \cup &
   		\{
   			\Active(0,1;1,0,\dots,0),
   			\dots,
   			\Active(0,1;0,\dots,0,1)
   		\}.
	\end{array}
\end{displaymath}

It is easy to check\footnote{This can be proven along the lines of
Claim~\ref{clm:induced-val}.} that from $V_1$ and the $V^i$ a satisfying
valuation for $\qr'_1$ over $I'$ can be constructed.
On the other hand, any
satisfying valuation of $\qr'_2$ over $I'$ would require to use facts
from $\alpha(1,I)$ for at least one
$\alpha(z_i, {\qr_2^i})$ and would thus induce a valuation of
$\qr^i_2$ over $I$, the desired contradiction.
\end{proof}

\section{Proofs for Section~\ref{sec:negation}: Parallel-correctness: unions of conjunctive queries with negation}

\begin{proof}[Proof of Proposition~\ref{prop:paracomp-upper} (continued)]
It only remains to describe the algorithm that tests the complement of parallel
completeness non-deterministically. On
input $\qr$ and $\distp$ (specified by $(n,T)\in\classpnondet^k$), the algorithm simply guesses an instance $J$
over a domain with at most $\varmax{\qr}$ values from $\dom_n$, and verifies that $J$ is a counter-example
showing that \qr is not parallel-complete under \distp. From
Claim~\ref{claim:smallmodel} it follows that this algorithm is
correct, since a counter-example must exist if  \qr is not
parallel-complete under \distp, and the actual data values do not matter.
It remains to show the complexity bounds and, in particular, to
describe how the verification part can  done.

In the general case, without a bound on the arity of the schema, the
verification is done as follows. The algorithm guesses a valuation $V$
that produces some fact \fc globally and which is not derived at any
node. To test that \fc is not derived at any node, the algorithm
cycles through all nodes and all valuations $V$ over $\dom_n$. The
number of combinations is bounded\footnote{We assume a binary
  alphabet, here.} by
$2^n\times(2^n)^{\varmax{\qr}}=2^{n(\varmax{\qr}+1)}$. Each test can
be performed by a simulation of all runs of $T$ which amounts to at
most $2^{n^k}$ simulations of at most $n^k$ steps each. Altogether the
algorithm needs time at most $2^{|(\qr,\distp)|^{k+2}}$.

If there is a fixed bound $\ell$ on the arity of the underlying schema
then the maximum size of the minimal counter-example becomes
polynomial and the test that \fc is derived globally can be done
non-deterministically in polynomial time and the test that it is not
derived locally can be done universally in polynomial time, thus
altogether yielding a \pitwo-computation.

The case of parallel-soundness is completely analogous (using the
second statement of Claim~\ref{claim:smallmodel} and the case of
parallel-correctness follows since it suffices to test parallel
completeness and soundness).
\end{proof}

\begin{proof}[Proof of Proposition~\ref{prop:paracomp-lower}
  (continued)]
 It remains to show that the described function is a reduction
from $\containment(\classBCQneg,\classBCQneg)$ to all of $\paracom(\classCQneg,\classpenum)$,  $\parasound(\classCQneg,\classpenum)$,
         and $\paracor(\classCQneg,\classpenum)$.

	To this end, we show first that containment $\qr_1 \subseteq
        \qr_2$ implies parallel-completeness and  
	parallel-soundness of query~$\qr$ under policy~$\distp$, and
        that lack of containment, $\qr_1 \not\subseteq \qr_2$, implies
	that query~$\qr$ is neither parallel-complete nor parallel-sound under
	policy~$\distp$. 

In both directions, we will make use of the following easy observations.

	\begin{claim}
		\label{clm:induced-val}
		Let $I$ be an arbitrary instance and $i \in \{1,2\}$.
		\begin{enumerate}
			\item Let $V$ be a valuation for $\qr$, let
                          $a\mydef V(\ell_i)$ and $V_i$ be
			the restriction of valuation~$V$ to variables in $\qr_i$. If $V$ satisfies
			$\qr$ on $I$, then $V_i$ satisfies $\qr_i$ on
			~$\alphainv{a}(I)$.
			
			\item Let $V_1,V_2$ be valuations for queries
                          $\qr_1,\qr_2$. Let $a,b$ be data values such
                          that $V_1$ and $V_2$ satisfy $\qr_1$ and
                          $\qr_2$, respectively, on $\alphainv{a}(I)$,
                          and such that
                          $\Type(b),\StartA(a),\StartB(a) \in I$, and
                          $\Stop(a) \notin I$. 
Then the valuation
			$W^b_a$, that agrees with $V_1$ and $V_2$ on all variables in
			$\qr_1$ and $\qr_2$, respectively, and maps $\ell_1,\ell_2 \mapsto a$,
			and $t \mapsto b$, satisfies $\qr$ on
			$I$.
		\end{enumerate}
	\end{claim}

\medskip
\noindent
	Let us now assume $\qr_1 \subseteq \qr_2$. To show that $\qr$ is
	parallel-complete under $\distp$, let $V$ be a valuation that
	globally satisfies $\qr$ on some arbitrary instance~$I$ over $U$. Let $a \mydef
	V(\ell_1)$ and $b \mydef V(t)$. Satisfaction of $\qr$ on $I$ by $V$ then
	particularly implies $\Type(b),\StartA(a),\StartB(a) \in I$ and $\Stop(a)
	\notin I$.
	Let $V_1$ be the satisfying valuation for $\qr_1$ on instance
        $\alphainv{a}(I)$, as given by Claim~\ref{clm:induced-val}.1. Since $\qr_1 \subseteq \qr_2$, there exists a valuation~$V_2$ that
	satisfies $\qr_2$ on instance $\delayer{a}(I)$. By
	Claim~\ref{clm:induced-val}.2, valuation~$W\mydef W^b_a$ satisfies $\qr$ on
	$\pi_a(I) \cup \{\Type(b),\StartA(a),\StartB(a)\}$.
	If $W(t)=1$, then node~$\nodeComp_a$ is responsible for these facts; if
	$W(t)=2$, then node~$\nodeSound_a$ is responsible for these facts; and
	otherwise node~$\nodeRest$ is responsible for these facts. Hence, query~$\qr$
	is parallel-complete under policy~$\distp$.
	
The proof that $\qr$ is parallel-sound under $\distp$ is similar. To
this end,
	let $I$ be a global instance and  $V$ be a valuation that satisfies $\qr$ on
	the local instance~$I_\node$ of some node $\node \in \nw$. Let $a \mydef
	V(\ell_1)$ and $b \mydef V(t)$. By definition of $\distp$, we can infer
	$\Type(b),\StartA(a),\StartB(a) \in I$ and $\Stop(a) \notin I$ from
	satisfaction of $\qr$ by $V$. 	Let $V_1$ be the satisfying valuation for $\qr_1$ on instance
        $\alphainv{a}(I_\node)$, as given by Claim~\ref{clm:induced-val}.1. 
	Since
	$\qr_1 \subseteq \qr_2$, there exists a valuation~$V_2$ that satisfies $\qr_2$
	on instance $\delayer{a}(I_\node)$. By definition
	of $\distp$, we have $\delayer{a}(I_\node)=\delayer{a}(I)$,
        and therefore $V_2$ satisfies $\qr_2$
	on $\delayer{a}(I)$. Thus, by
	Claim~\ref{clm:induced-val}.2, valuation~$W\mydef W^b_a$ satisfies $\qr$ on
	$I$. Hence, query~$\qr$ is parallel-sound under
	policy~$\distp$.
	
	\medskip
	\noindent
	Let now $\qr_1 \not\subseteq \qr_2$ and let $I_1$ be an instance such that $\qr_1(I_1) \neq
	\emptyset$ and $\qr_2(I_1) = \emptyset$. Thanks to
        Lemma~\ref{lem:reduction_helper}, we can assume that $I_1$ is
        over a domain of size at most $\varmax{\qr_1}$. Thanks to
        genericity, we can assume that the domain is a subset of $U$.
Let $V_1$ be a valuation that
	satisfies query~$\qr_1$ on $I_1$. Furthermore, let $V_2$ be some consistent
	valuation for $\qr_2$ and $I_2 \mydef V_2(\bodypos{\qr_2})$,
        which exists thanks to our assumption that $\qr_2$ is satisfiable.
	
	To show that $\qr$ is not parallel-complete under $\distp$, we
	define $I \mydef \layer{1}(I_1) \cup \layer{2}(I_2) \cup
	\{\Type(1),\StartA(1),\StartB(1),\StartB(2)\}$. Then, valuation $V$ which maps
	all variables in $\qr_1$ and $\qr_2$ as $V_1$ and $V_2$, respectively, and
	$\ell_1 \mapsto 1$, $\ell_2 \mapsto 2$, and $t \mapsto 1$ satisfies query~$\qr$
	on instance~$I$.
	However, there is no locally satisfying valuation for $\qr$. For a
	contradiction, assume existence of such a valuation~$W$. Since $\Type(1)$
	and $\StartA(1)$ are the only $\Type$- and $\StartA$-facts contained in the
	global instance, we have $W(t)=W(\ell_1)=1$. By definition of $\distp$, this
	valuation can only be satisfying on node~$\nodeComp_1$. Since only $\StartB(1) \in
	\dist{I}(\nodeComp_1)$, Claim~\ref{clm:induced-val}.1 implies existence of a
	satisfying valuation $W_2$ for $\qr_2$ on $\delayer{1}(I) = I_1$, which
	contradicts the choice of instance~$I_1$. Hence, query~$\qr$ is not
	parallel-complete under policy~$\distp$.
	
	To show that $\qr$ is also not parallel-sound under $\distp$,
	let instance $I \mydef \layer{1}(I_1) \cup \layer{2}(I_2) \cup \{\Type(2),
	\StartA(1), \StartB(1), \StartB(2), \Stop(2)\}$. Then, valuation $V$ which maps
	all variables in $\qr_1$ and $\qr_2$ as $V_1$ and $V_2$, respectively, and
	$\ell_1 \mapsto 1$, $\ell_2 \mapsto 2$, and $t \mapsto 2$ satisfies
	query~$\qr$ on the local instance $\dist{I}(\nodeSound_1)$ of
	node~$\nodeSound_1$ because this node is not responsible for fact $\Stop(2)$.
	However, there is no globally satisfying valuation for $\qr$. Towards a
	contradiction, assume existence of such a valuation~$W$. Since
	$\StartB(1),\StartB(2)$ are the only $\StartB$-facts in the global instance, we
	have $W(\ell_2)=1$ or $W(\ell_2)=2$. The latter cannot hold because the
	valuation then prohibits the present fact $\Stop(2)$. The former implies, by
	Claim~\ref{clm:induced-val}.1, a satisfying valuation $W_2$ for $\qr_2$ on
	$\delayer{1}(I)=I_1$, which contradicts the choice of instance~$I_1$. Hence,
	query~$\qr$ is not parallel-sound under policy~$\distp$.

This completes the proof that problem
$\containment(\classBCQneg,\classBCQneg)$ is reducible to all three problems
$\paracom(\classCQneg,\classpenum)$,  $\parasound(\classCQneg,\classpenum)$, and
$\paracor(\classCQneg,\classpenum)$ in polynomial time and thus shows
$\coNEXPTIME$-hardness via Corollary~\ref{coro:cqneg-hard}.
\end{proof}

\section{Proofs for Section~\ref{sec:full}: Full conjunctive queries}

\newcommand{\subQcollapse}{\EM{\qr^{\forall h}}}
\newcommand{\subQeqineq}{\EM{\qr^{\pred{Neq}}}}
\newcommand{\subQfour}{\EM{\qr^{>3}}}

\newcommand{\subQwronginequality}{\EM{\qr^{\pred{Neq}}_{\text{equal}}}}
\newcommand{\subQinconsistent}{\EM{\qr^{\pred{Eq},\pred{Neq}}_{\text{amb}}}}
\newcommand{\subQeqasymmetry}{\EM{\qr^{\pred{Neq}}_{\text{asym}}}}
\newcommand{\subQineqasymmetry}{\EM{\qr^{\pred{Eq}}_{\text{asym}}}}
\newcommand{\subQincomplete}{\EM{\qr^{\pred{Eq},\pred{Neq}}_{\text{undef}}}}

In the proofs below we drop the convention that in \UCQ{}s all disjuncts have pairwise
disjoint variable sets (which was introduced to simplify the proofs in other sections). However, one can easily observe that a \UCQ that does not comply with
this convention can be easily transformed to one that does, by for example, adding
distinct indices to variables in separate disjuncts.

\subsection{Proof for Theorem~\ref{theo:full_containment_new}}

Theorem~\ref{theo:full_containment_new} follows from the hardness results in
Proposition~\ref{pro:full_containment_hardness}, and the upper bound in
Proposition~\ref{prop:containment_conpcomplete}, which are given below.

\begin{proposition}
	\label{the:full_containment_hardness}\label{pro:full_containment_hardness}
    \begin{enumerate}
        \item $\text{\sc Containment}(\classFCQnegineq, \classFCQnegineq)$ is in $\Poly$;
        \item 
    $\text{\sc Containment}(\classFCQnegineq, \classFUCQnegineq)$ is \coNP-hard.
    \end{enumerate}
\end{proposition}

\begin{proof}
    \textbf{1.} 
     Let $\qr_1, \qr_2 \in \classFCQnegineq$. We show that  $\qr_1 \subseteq \qr_2$
     if and only if $h(\pos{\qr_2}) \subseteq \pos{\qr_1}$ and
     $h(\bodyneg{\qr_2}) \subseteq \bodyneg{\qr_1}$ for the substitution $h$
     that identifies the head relation,\footnote{In particular, containment
     implies existence of such a substitution despite possible multiple
     occurences of the same variable in $\head{\qr_2}$, as the following proof
     shows.
 }
 $h(\head{\qr_2}) =
     \head{\qr_1}$.
     It then immediately follows that
     $\containment(\classFCQnegineq, \classFCQnegineq)$ is in $\Poly$;

To show the claim, let $\qr_1\subseteq \qr_2$. 
Let $I^{-}$ be the minimal canonical database for $\qr_1$, i.e., $I^{-}$ consists of the
frozen atoms in $\pos{\qr_1}$. By $I^{+}$ we denote the maximal canonical database for
$\qr_1$, i.e., $I^{+}$ contains every frozen atom over $\vars{\qr}$ (and the relations where
$\qr_1$ and $\qr_2$ are defined over) that is not in $\negbody{\qr_1}$. 
By construction, 
$\head{\qr_1} \in \qr_1(I^{-})$ and $\head{\qr_1} \in \qr_1(I^{+})$, which implies by
containment that $\head{\qr_1} \in \qr_2(I^{-})$ and $\head{\qr_1} \in \qr_2(I^{+})$. By
fullness of $\qr_2$, both are derived by the same valuation $V$. 
Now, the former implies $V(\pos{\qr_2}) \subseteq I^{-} =
\pos{\qr_1}$, and the latter implies $V(\negbody{\qr_2})\cap I^+ = \emptyset$. Thus,
$V(\negbody{\qr_2}) \subseteq \negbody{\qr_1}$. Hence, $V$ describes the desired substitution.

For the other direction, suppose that substitution $h$ has the desired properties.
Let $I$ be an arbitrary instance, and $\fc \in \qr_1(I)$. Thus, there is a valuation
$V_1$, where $V_1(\head{\qr_1}) = \fc$, $V_1(\pos{\qr_1}) \subseteq I$, and
$V_1(\negbody{\qr_1}) \cap I =
\emptyset$. Let $V_2 \mydef V_1\circ h$. Then, 
$V_2(\head{\qr_2}) = V_1(\head{\qr_1}) = \fc$, $V_2(\pos{\qr_2}) \subseteq V_1(\pos{\qr_1})
\subseteq I$, and $V_2(\negbody{\qr_2}) \subseteq V_1(\negbody{\qr_1})$.
The latter implies $V_2(\negbody{\qr_2}) \cap I = \emptyset$, and thus $\fc \in \qr_2(I)$.


    \medskip \noindent \textbf{2.} 
    The reduction is from the graph $3$-colorability problem, which is well-known to be
    \NP-complete, and asks for a given graph $G$ whether there is a coloring of the nodes
    in $G$, using only $3$ colors, such that adjacent nodes have distinct colors.  Let $G$
    be an arbitrary input for the described problem. 

    We construct queries $\qr_1\in\classFCQnegineq$ and $\qr_2 \in \classFUCQnegineq$, and show 
    $\qr_1 \subseteq \qr_2$ if and only if $G$ is \emph{not} $3$-colorable.
    Intuitively, this is done by letting $\qr_1$ derive tuples that each represent a
    complete labelling function for (a substitution of) graph $G$. Semantically,
    query $\qr_2$ is very similar to 
    $\qr_1$, but derives only invalid labelling functions, that is, which either
    give two adjacent nodes the same color, or use more than three colors. 
    The latter is implemented as a union of queries where each disjunct
    detects a particular issue.

    Queries $\qr_1$ and $\qr_2$ are defined over database schema $\sch \mydef \{
        \pred{E}^{(2)}, \pred{L}^{(2)}\}$, where $\pred{E}$ represents the edges of $G$,
        and $\pred{L}$ is used to model mappings from nodes in $G$ onto colors.

Before going to the construction itself, we first define $\text{\sc edges}$ as the set of
atoms describing the edges of $G$. We do this by taking for every node in $G$ a unique variable in $\uvar$,
and adding the atom $\pred{E}(x,y)$ to $\text{\sc edges}$ if and only if the nodes
represented by $x$ and $y$ are adjacent in $G$.
Second, we define the set $\text{\sc labels}$, consisting of atoms $\pred{L}(x, y_x)$
for every node-representing variable $x$, where $y_x$ denotes a fresh variable. We call
$y_x$ a color-representing variable.

Henceforth, we denote by $\vec{x}$ the node-representing variables, and by $\vec{y}$ the
color-representing variables. Both are assumed to have a fixed order.

Now, we define query $\qr_1$ as:
\begin{align*}
    H(\vec{x}, \vec{y}) \leftarrow \text{\sc edges},\text{\sc labels}.
\end{align*}
Notice that, on given instance $I$, $\qr_1$ outputs every possible labelling
function for substitutions of graph $G$ representable with the facts in $I$.

Query $\qr_2$ is slightly more complex, therefore we chop $\qr_2$ down in 
two semantically-meaningful $\FUCQneg$s: 
\subQcollapse, which detects whether a described labelling function assigns the
same color to some pair of adjacent nodes; and \subQfour, which detects whether a
described labelling function uses more than three colors. 

We describe the sketched queries in more detail.
To this end, let $x_1, x_2 \in \vec{x}$ be two distinct variables, representing adjacent
nodes in $G$. Based on $x_1$ and $x_2$, we can define a substitution $h$
mapping all the variables of $\vec{x}$ and $\vec{y}$ onto themselves, except for $y_{x_1}$ and $y_{x_2}$, which are mapped onto a fresh variable $y$.
We call $h$ a \emph{collision-revealing substitution} for $\vec{x}$ and $\vec{y}$.

Now $\subQcollapse$ is the query defined as the union of queries:
\begin{align*}
    h(\head{\qr_1}) \leftarrow h(\body{\qr_1}),
\end{align*}
for all collision-revealing substitutions $h$ for $\vec{x}$ and $\vec{y}$ as described
above.

Query $\subQcollapse$ outputs exactly those labelling functions outputted by
$\qr_1$ that assign a same color to at least some pair of adjacent nodes.

Query $\subQfour$ is defined as the union of queries:
\begin{align*}
    \head{\qr_1} \leftarrow \body{\qr_1}, \forall_{i<j} y_i \ne  y_j,
\end{align*}
for every combination of distinct variables $y_1, y_2, y_3, y_4$ in $\vec{y}$.

As the construction of the individual $\FCQneg$ can be done in polynomial time in the
size of $G$, and there are only $n^2$ many collision-revealing substitutions, the construction
of $\qr_1$ and $\qr_2$ can be done in polynomial time in the size of $G$. 

\smallskip
	\noindent
	\emph{Correctness.}
We show that $\qr_1 \not\subseteq \qr_2$ if
and only if $G$ is $3$-colorable.

\IfDirection
Suppose $G$ is $3$-colorable. So, there is a labelling function $\ell$ mapping the nodes
in $G$ onto $3$-colors, such that no two adjacent nodes are assigned the same color.
We abuse notation and assume nodes in $G$, as well as the colors in the image of
$\ell$, are over $\dom$. 
Let $I \mydef \{ \pred{E}(n_1,n_2) \mid \text{$(n_1,n_2)$ is an edge in $G$}\} 
\cup \{ \pred{L}\big(n,\ell(n)\big) \mid \text{$n$ is a node in $G$}\}$.

Now, let $V_1$ be the valuation mapping the node-representing variables in $\qr_1$ onto the nodes of $G$ that they
represent, and the color-representing variables onto colors as defined by $\ell$. Obviously,
$V_1$ satisfies on $I$ and derives a fact $\fc$. 

By choice of $\ell$, $\fc \not \in \subQcollapse(I)$, because none of its
adjacent nodes are labelled the same color. 
Further, $\ell$ introduces only three colors, and thus $\fc \not \in \subQfour(I)$, which
implies $\fc \not \in \qr_2(I)$.
Hence, $\qr_1 \not\subseteq \qr_2$.

\OnlyIfDirection
Suppose $\qr_1\not\subseteq \qr_2$. Thus, there is an instance $I$ and fact $\fc$, where
$\fc \in \qr_1(I)$, and $\fc\not \in\qr_2(I)$.  Let $V_1$ be the valuation for
$\qr_1$ on $I$.

The former implies that $I$
contains an interpretation of the described graph. Notice that 
$\fc$ does not necessarily describe $G$ itself, but rather a substitution of $G$ under
which some nodes might have been collapsed.
We consider the labelling function $\ell$, defined as $\ell(n) \mydef V_1(x)$,
for every $x \in \vec{x}$, where $x$ is the variable representing node $n$ in
$G$.

By fullness of the considered queries, for each of the disjuncts of $\qr_2$, there is only
one valuation that can possibly derive $\fc$, and $\fc$ uniquely defines this valuation
for the respective disjunct.
Therefore, from $\fc \not \in \subQcollapse(I)$ we directly obtain 
that for every pair of adjacent nodes $n_1,n_2$ in $G$: $\ell(n_1) \ne \ell(n_2)$.
So, $\ell$ describes a valid labelling function for $G$. 
It remains to argue that $\ell$ does not use more than three colors.
The latter follows from $\fc \not \in \subQfour(I)$, 
and thus the substitution of $G$ described by $\fc$ must be $3$-colorable.
From this it immediately follows that $G$ must be $3$-colorable as well.
\end{proof}

We next show that the problem remains \coNP-hard, even for queries
without inequalities.

\begin{lemma}
   $\containment(\classFCQnegineq,
    \classFUCQnegineq)\le_p\containment(\classFCQneg,
    \classFUCQneg)$.
\end{lemma}

\begin{proof}
For this to see, let $\qr_1 \in \classFCQnegineq$ and $\qr_2 \in
\classFUCQnegineq$ over some schema $\sch$. Let $\vec{y}$ represent the variables
used in $\qr_1$, in some fixed order.
We construct a query $\qr_2' \in \classFUCQneg$ such that $\qr_1 \subseteq \qr_2$ if and
only if $\qr_1 \subseteq \qr_2'$. 
For this, we extend $\sch$ to a schema $\sch'$ that also defines relations
$\pred{Eq}^{(2)}$ and $\pred{Neq}^{(2)}$.
Intuitively, $\pred{Eq}^{(2)}$ models the equality relation $=$, and $\pred{Neq}$ models
the inequality relation $\ne$.
We assume $\pred{Eq}$ and $\pred{Neq}$ are not in $\sch$. 

For the construction, we divide $\qr_2'$ into two subqueries $\qr^\ast_2$ and
$\subQeqineq$. Query~$\qr^\ast_2$ results from $\qr_2$ by replacing every
inequality $x \neq y$ by atom $\pred{Eq}(x,y)$. The purpose of query
$\subQeqineq$ is to allow derivation of a fact derivable by $\qr_1$ on every
instance where $\pred{Eq}$ and $\pred{Neq}$ do not represent equality
relation~$=$ or inequality relation~$\neq$, respectively.

To this end,
we further divide $\subQeqineq$ into the following queries:
\subQwronginequality, which detects values that are wrongly identified as
being unequal (by $\pred{Neq}$);
\subQinconsistent, which detects that some values occur in both $\pred{Eq}$ and $\pred{Neq}$;
\subQineqasymmetry, which detects values where $\pred{Neq}$ is not symmetric for;
\subQeqasymmetry, which detects values where $\pred{Eq}$ is not symmetric for;
\subQincomplete, which detects that certain values are not in $\pred{Neq}$, nor in
$\pred{Eq}$; and finally, $\subQeqineq$, in which all occurrences of $x\ne y$ are replaced by
atoms $\pred{Neq}(x,y)$.

More formally, $\subQwronginequality$ is defined as the union of the queries:
\begin{align*}
    h(\head{\qr_1}) \leftarrow h(\body{\qr_1}), \pred{Neq}(y,y),
\end{align*}
for all collision-revealing substitutions $h$ for $\vec{x}$ and $\vec{y}$.

Thus, $\subQwronginequality$ outputs on a given instance exactly those facts outputed by
$\qr_1$ in which some values are wrongly identified as being inequal.

We define query $\subQinconsistent$, as the union over queries:
\begin{align*}
    \head{\qr_1} \leftarrow \body{\qr_1}, \pred{Eq}(y_1, y_2), \pred{Neq}(y_1, y_2),
\end{align*}
for every two distinct variables $y_1,y_2 \in \vec{y}$.
Query  $\subQineqasymmetry$ is defined as the union over queries:
\begin{align*}
    \head{\qr_1} \leftarrow \body{\qr_1}, \pred{Eq}(y_1, y_2), \lnot\pred{Eq}(y_2, y_1),
\end{align*}
for every two distinct variables $y_1,y_2 \in \vec{y}$.
Query $\subQeqasymmetry$ is defined as the union over queries:
\begin{align*}
    \head{\qr_1} \leftarrow \body{\qr_1}, \pred{Neq}(y_1, y_2), \lnot\pred{Neq}(y_2, y_1).
\end{align*}
for every two distinct variables $y_1,y_2 \in \vec{y}$.
And, eventually,  query $\subQincomplete$ is defined as the union over queries:
\begin{align*}
    \head{\qr_1} \leftarrow \body{\qr_1}, \lnot\pred{Eq}(y_1, y_2), \lnot\pred{Neq}(y_1,
    y_2),
\end{align*}
for every two distinct variables $y_1,y_2 \in \vec{y}$.

    \smallskip \noindent \emph{Correctness.}
    It remains to show $\qr_1 \subseteq \qr_2$ if
    and only if $\qr_1 \subseteq \qr_2'$. 

\IfDirection
Suppose $\qr_1 \not\subseteq \qr_2$. Thus, for some instance $I$ and fact $\fc$, $\fc
\in \qr_1(I)$, while $\fc \not \in \qr_2(I)$.  We consider instance $I'$ over $\sch'$,
which consists of all facts in
$I$, and for every value $a \in\adom{I}$ a fact $\pred{Eq}(a,a)$, and for every pair of
distinct values $a,b \in \adom{I}$ a fact $\pred{Neq}(a,b)$.

Because $\qr_1$ does not reference relations $\pred{Eq}$ or $\pred{Neq}$, it follows that $\fc \in \qr_1(I')$. 
As relations $\pred{Eq}$ and $\pred{Neq}$ express exactly $=$ and $\ne$, over the values
in $\adom{I}$, we obtain $\qr_2(I) = \qr_2^*(I')$.
Consequently, $\fc \not \in \qr_2(I)$ implies $\fc \not \in \qr_2^*(I')$.
Further, as $\pred{Eq}$ and $\pred{Neq}$ are symmetric, consistent, and completely defined over
$\adom{I'}$, $\subQeqineq(I') = \emptyset$.
Hence, $\fc \not \in \qr_2'(I')$. 

\OnlyIfDirection
Suppose $\qr_1 \not\subseteq \qr_2'$. Thus, for some instance $I$ and fact $\fc$, $\fc \in
\qr_1(I)$, and $\fc \not \in \qr_2'(I)$.  Let $V$ be the valuation deriving $\fc$ for
$\qr_1$ on $I$ and $D \mydef \adom{V(\body{\qr_1})}$. 
From $\fc \not \in \qr_2'(I)$ it now follows that $\pred{Eq}$ and $\pred{Neq}$ are
well-defined over $D$, that is, the relations $\pred{Eq}$ and $\pred{Neq}$ are
identical to $=$ and $\ne$ over values in $D$. 
Indeed, for all $a,b \in D$:
\begin{itemize}
    \item there are no facts $\pred{Neq}(a,a)$
        (because
        $\fc \not \in \subQwronginequality(I)$);
    \item $\pred{Eq}(a, b)$ in $I$ implies $\pred{Eq}(b,a)$ (because $\fc \not \in
        \subQeqasymmetry(I)$);
    \item $\pred{Neq}(a,b)$ in $I$ implies $\pred{Neq}(b,a)$
        (because $\fc \not \in
        \subQineqasymmetry(I)$); and
    \item there is a fact $\pred{Eq}(a, b)$ or $\pred{Neq}(a,b)$ in $I$ (from $\fc \not \in \subQincomplete(I)$).
\end{itemize}

As a result, we can simply replace every occurrence of $\pred{Neq}(x,y)$ in
query $\qr_2^*$ by $x\ne y$ without changing its semantics over $D$.
More formally, we obtain $\qr_2^*(\restrict{I}{D}) = \qr_2(\restrict{I}{D})$. 

By fullness of $\qr_2$ and the definition of $D$, $\fc \in \qr_2(I)$ would imply
$\fc\in \qr_2(\restrict{I}{D})$ and thus $\fc \in \qr'_2(I)$. Consequently, it must be that $\fc
 \not \in \qr_2(I)$.
\end{proof}

\begin{proposition}
	\label{prop:containment_conpcomplete}
    $\text{\sc Containment}(\classUFCQnegineq, \classFUCQnegineq)$ is in \coNP.
\end{proposition}
\begin{proof}
    We observe that when $\qr\not\subseteq\qr'$, there is an instance $I$ and
    fact $\fc$, such that $\fc \in \qr(I)$, and $\fc \not \in \qr'(I)$. 
    In particular, by fullness of $\qr_1$ and $\qr_2$, there is such an instance of size
    at most $\ell \mydef \max_{i\in\{1, \ldots, n\}}\{\ssize{\bodypos{\qr_i}}\}
    + m$, where $n$ denotes the number of disjuncts of $\qr$, and $m$
    the number of disjuncts in $\qr'$.

    Indeed, to see this, let $V$ and $\qr_i$ be the valuation and disjunct of $\qr$ where $\fc$ is
    derived by on $I$. Let $J \mydef V(\pos{\qr_i})$.
    By fullness it follows that for each disjunct $\qr_j$ of $\qr'$ there is at most one
    valuation $V_j$ eligible to derive $\fc$ for $\qr_j$. As, by choice of $\fc$, $V_j$ does
    not satisfy on $I$ it must be that either $V_j(\pos{\qr_j}) \not\subseteq I$, or
    $V_j(\negbody{\qr_j})\cap I \ne \emptyset$. We ignore the former. In the latter case
    we choose one fact from  $V_j(\negbody{\qr_j})\cap I$ and add it to $J$. 
    One can now easily verify that $J$ is as desired.

    As the above shows that there always is a witnessing instance $I$ of size $\ell$,
    given $I$, $V$, and $\qr_i$ as polynomial size certificate, we can easily verify that
    $\qr$ is indeed not contained in $\qr'$, by simply verifying that $V$ indeed satisfies
    for disjunct $\qr_i$ of $\qr$ on $I$, and for all $m$ eligible valuations for
    conjuncts of $\qr'$, either not all required facts are in $I$, or at least one of
    the prohibited facts is present.
\end{proof}

\subsection{Proof of Proposition~\ref{the:full_parallelXreductions:new}}

The result follows from Proposition~\ref{pro:full_containment_hardness} and the reductions below.

\begin{proposition}
	\label{the:full_parallelXreductions}
	The following reducibility relations hold, for every $\classp \in
	\{\classpenum\} \cup \Classpnondet$:
	\begin{enumerate}
		\item $\containment(\classFCQneg,\classUFCQneg) \polred \parasound(\classUFCQneg,
            \classp)$;
		\item $\containment(\classFCQneg,\classUFCQneg) \polred \paracom(\classUFCQneg,
            \classp)$;
		\item $\containment(\classFCQneg,\classUFCQneg)
                  \polred \paracor(\classUFCQneg,\classp)$. 
	\end{enumerate}
\end{proposition}

\begin{proof}
	The idea underlying the following reductions is simple: extend both
	$\CQneg$s $\qr$ and $\qr'$ for the containment problem by a nullary
	atom~$\Global()$ or its negation $\lnot\Global()$, respectively, and combine
	them (by union) into a single query~$\qr^\ast$. By mapping the fact $\Global()$ onto an isolated node, the
	distribution policy then allows to control global and local derivability on
	behalf of $\qr$ and $\qr'$. Without loss of generality, we always assume that queries $\qr$ and $\qr'$ do not use the auxiliary relation $\Global$.

	For a $\CQneg$ $\qr$, let $\qrglobal$ and
	$\qrnotglobal$ denote the queries obtained by adding the
        literal $\Global()$ or $\lnot\Global()$ to $\qr$,
	respectively.
    For unions of $\CQneg{}$s this particularly means adding $\Global()$ or
    $\lnot\Global()$ to every disjunct of $\qr$. 
	The following identities can be easily proven to hold for every query~$\qr$
	and every instance~$I$:
	\begin{subequations}
        \begin{align}
			\qrglobal(I \cup \{\Global()\}) & =  \qr(I); \label{eqn:glA}\\
            \qrglobal(I \setminus \{\Global()\}) & =  \emptyset; \label{eqn:glB}\\
			\qrnotglobal(I \cup \{\Global()\}) & =  \emptyset; \label{eqn:glC}\\
		    \qrnotglobal(I \setminus \{\Global()\}) & =  \qr(I).\label{eqn:glD}
        \end{align}
	\end{subequations}

	\medskip
	\noindent
	We only argue the reductions for policy class~$\classpenum$. It is obvious that
	such a policy can also be represented by a policy from any class
	$\classpnondet^k \in \Classpnondet$.

	\noindent
	\textbf{1.}
	We start with $\containment(\classFCQneg,\classFCQneg) \polred
	\parasound(\classUFCQneg, \classpenum)$.

    Let $\qr$ and $\qr'$ be $\CQneg$s. We
	define a $\UCQneg$ $\qr^\ast$ and a policy~$\distp$ as
	follows. For this, let $\qr^\ast \mydef \qrnotglobal \cup \qrBglobal$. 
    
    We construct a subset $D$ of $\dom$, where $\ssize{D}$ = $\ssize{\vars{\qr_1}}$.
    Since the actual data values do not matter, we can choose those with
    the shortest representation length and thus also represent set~$D$
    polynomially in the size of query~$\qr_1$.
    Now, $\distp$ is defined as a distribution policy over network $\nw=\{\node_1, \node_2\}$
	that forwards every fact over $D$ except $\Global()$ to $\node_1$, and $\Global()$ to
    node $\node_2$.
    As the described distribution policy can be straightforwardly expressed with a
    distribution policy in $\classpfin$, the construction of both $\qr^\ast$ and
    $\distp$ can be done in polynomial time.
	
	\smallskip
	\noindent
	\emph{Correctness.}
	It remains to show that $\qr \subseteq \qr'$ if and only if $\qr^\ast$ is
    parallel-sound under $\distp$. 	
    The following observations are crucial for the correctness argument. 
    
    First of all, for
	each instance~$I$ that contains $\Global()$, $\qr^\ast$
	is equivalent to $\qr'$:
	\begin{equation}
		\label{eqn:redtoqrB}
		\qr^\ast(I)
		=
		\qrnotglobal(I) \cup \qrBglobal(I)
		=
		\emptyset \cup \qr'(I)
		=
		\qr'(I).
	\end{equation}
    Second, for each instance~$I$ that does not contain
	$\Global()$, $\qr^\ast$ is equivalent to $\qr$:
	\begin{equation}
		\label{eqn:redtoqr}
		\qr^\ast(I)
		=
		\qrnotglobal(I) \cup  \qrBglobal(I)
		=
		\qr(I) \cup \emptyset
		= \qr(I) = \qr\big(I \cup \{\Global()\}\big).
	\end{equation}

    Further, as $\node_2$ can only contain the fact $\pred{Global}()$, it follows that for
    each instance $I$ we have 
    $\qr(\dist{I}(\kappa_2)) =
    \emptyset$.

    \OnlyIfDirection
    Assume $\qr \subseteq \qr'$. Let $I$ be an arbitrary subset of $\facts{\distp}$. 
        If $\Global()
	\notin I$, the local instance of node~$\node_1$ is identical to the
	global instance, that is $\dist{I}(\node_1)=I$. Now, by definition of $\distp$ and
    thus also 	the result sets, $\qr^\ast\big(\dist{I}(\node_1)\big) = \qr^\ast(I)$. 
    In particular, this implies $\qr^\ast\big(\dist{I}(\node_1)\big) \subseteq
	\qr^\ast(I)$, that is, parallel-soundness of $\qr^\ast$ under $\distp$ on
	instance~$I$. 
	
	If $\Global() \in I$, the local instance is 
	$\dist{I}(\node_1)=I \setminus \{\Global()\}$. By
	Equations~\eqref{eqn:redtoqrB} and \eqref{eqn:redtoqr},
	$\qr^\ast\big(\dist{I}(\node_1)\big) \subseteq \qr^\ast(I)$ if and only if
	$\qr(I) \subseteq \qr'(I)$, which holds by assumption of containment.
	Therefore, in both cases query~$\qr^\ast$ is parallel-sound under policy~$\distp$.
	
	\IfDirection
	For a proof by contraposition assume $\qr \not\subseteq \qr'$. This implies
	existence of an instance~$I$ where $\qr(I) \not\subseteq \qr'(I)$.
    Without loss of generality we can assume 
    $\adom{I} \subseteq  D$.
    The latter is a safe assumption as from Lemma~\ref{lem:reduction_helper} it follows
    that an instance $J\subseteq I$ exists that preserves the desired property, and
    where $\ssize{\adom{J}} \le \ssize{D}$. Further, by genericity of $\qr_1$ and $\qr_2$,
    we can uniquely rename data values in $J$ to data values in $D$, resulting
    in an instance with the desired properties.

    Additionally, we may
	safely assume that $\Global() \in I$ because neither $\qr$ nor $\qr'$ refers to
	relation~$\Global$. This results in a local instance
	$\dist{I}(\node_1) = I \setminus \{\Global()\}$.
	Again, by Equations~\eqref{eqn:redtoqrB} and \eqref{eqn:redtoqr}, we conclude
	$\qr^\ast\big(\dist{I}(\node_1)\big) = \qr(I) \not\subseteq \qr'(I) =
	\qr^\ast(I)$, that is $\qr^\ast$ is not parallel-sound under $\distp$ on instance~$I$. Therefore, by contraposition,
	parallel-soundness of query~$\qr^\ast$ under policy~$\distp$ and domain $D$ implies
	containment $\qr \subseteq \qr'$.
	

    \medskip \noindent \textbf{2.} For $\containment(\classFCQneg,\classFCQneg) \polred
    \paracom(\classUFCQneg, \classpenum)$ the proof is analogous.
    Policy~$\distp$ is defined as before, while query $\qr^\ast \mydef \qrglobal \cup \qrBnotglobal$,  i.e.
    atoms $\Global()$ and $\lnot\Global()$ are swapped between the $\CQneg$s compared to
    the soundness reduction. 

    Analogously to Equations~\eqref{eqn:redtoqrB} and \eqref{eqn:redtoqr}, this leads to
    $\qr^\ast(I)=\qr(I)$ and $\qr^\ast(I)=\qr'(I \cup \{\Global()\})$, for
    every instance~$I$ where $\Global() \in I$ and $\Global() \notin I$, respectively.
    Correctness then follows again from Lemma~\ref{lem:reduction_helper}.


    \medskip
	\noindent
	\textbf{3.}
	For $\containment(\classFCQneg,\classFCQneg) \polred
	\paracor(\classUFCQneg,\classpenum)$ is also similar. Policy~$\distp$ is
	defined as before while query $\qr^\ast \mydef \qrnotglobal \cup \qr'$. If $\Global()\in I$ then $\qr^\ast(I)=\qr'(I)$ just as in the first
        reduction. However, if $\Global()\not\in I$
        then $\qr^\ast(I)=\qr(I)\cup\qr'(I)$.

	\OnlyIfDirection
	Assume $\qr \subseteq \qr'$. Thus, $\qr\cup \qr'
        \equiv \qr'$ imply $\qr^\ast \equiv \qr'$. Let $I$ be an arbitrary
        instance such that $\adom{I} \subseteq D$.
        If $\Global()
	\notin I$, the local instance of node~$\node_1$ is identical to the
	global instance (i.e., $\dist{I}(\node_1)=I$) by definition of $\distp$, and thus also
	$\qr^\ast\big(\dist{I}(\node_1)\big) = \qr^\ast(I)$. 
	
	If $\Global() \in I$ is 
	$\dist{I}(\node_1)=I \setminus \{\Global()\}$. Thus,
    $\qr^\ast\big(\dist{I}(\node_1)\big) = \qr'(\dist{I}(\node_1)) = \qr'(I) = \qr^\ast(I)$.
	Therefore, in both cases query~$\qr^\ast$ is parallel-correct under policy~$\distp$.
	
	\IfDirection
	We assume $\qr \not\subseteq \qr'$. This implies again
    by Lemma~\ref{lem:reduction_helper} and genericity of $\qr$ and $\qr'$, the existence of an instance~$I$ where $\qr(I)
    \not\subseteq \qr'(I)$ and $\adom{I} \subseteq D$. It follows $\qr(I)\cup \qr'(I)
        \not\equiv \qr'(I)$. We may
	safely assume that $\Global() \in I$ because neither $\qr$ nor $\qr'$ refers to
	relation~$\Global$. Therefore,
	$\dist{I}(\node_1) = I \setminus \{\Global()\}$.
	Again, by Equations~\eqref{eqn:redtoqrB} and \eqref{eqn:redtoqr}, we conclude
	$\qr^\ast\big(\dist{I}(\node_1)\big) = \qr(I) \cup \qr'(I) \not= \qr'(I) =
	\qr^\ast(I)$, that is, $\qr^\ast$ is not parallel-correct under
        $\distp$ on instance~$I$, the desired contradiction.
\end{proof}

\subsection{Proof for Theorem~\ref{the:parallelX_conpcomplete_full}}

\begin{proof}
    Let $\qr$ be a \FUCQnegineq with $m$ disjuncts, and $\distp = (U, \respp) \in \classpenum$.


    \medskip \noindent \textbf{1.}
From the definition of parallel-soundness we immediately obtain that
$\qr$ is not parallel-sound on $\distp$ if
and only if there is an instance $I\subseteq \facts{\sch, U}$, a fact $\fc$, and a node $\kappa$, where $\fc \in
\qr(\distvar{\distp}{I}(\kappa))$  and $\fc \not \in \qr(I)$.
The former implies that there is a valuation $V$ for some of the disjuncts of
$\qr$, say $\qr_i$, which derives $\fc$ on $\distvar{\distp}{I}(\kappa)$. The latter means that
for all disjuncts $\qr_j$, and valuations $V_j$ for $\qr_j$, where $V_j(\head{\qr_j}) =
\fc$, $V_j$ fails to satisfy for $\qr_j$ on $I$.

As $\qr$ is full, for each disjunct $\qr_j$ there is at most
one eligible valuation $V_j$, which is uniquely defined by the head of $\qr_j$ and the
desired output fact $\fc$. We thus have to check only $m$ valuations $V_j$. 
From this observation it also follows that we can
restrict the above condition to instances $I$ of size at most $\max_{j \in
    \{1,\ldots,m\}}\{\ssize{\bodypos{\qr_j}}\} +
m$. Indeed, given $I$, $\kappa$, and valuation $V$ (being the valuation showing $\fc \in
\qr(\distvar{\distp}{I}(\kappa))$), obviously $V(\pos{\qr_i}) \subseteq I$, and, for every $j\in\{1,
\ldots, m\}$, either $V_j(\pos{\qr_j})\not\subseteq I$ or there are facts in
$V_j(\negbody{\qr_j})\cap I$. In the latter case, let $\fcB_j$ be one of these facts.
Now, we can construct an instance $I'$ of the desired size by simply taking the facts in
$V(\pos{\qr_i})$ and the chosen facts $\fcB_j$. It is easy to see that $\fc\in
\qr\big(\distvar{\distp}{I'}(\kappa)\big)$, while $\fc \not \in \qr(I')$. 

As $V(\pos{\qr_i}) \subseteq \respvar{\distp}{\kappa}$ and $V(\negbody{\qr_i}) \cap \respvar{\distp}{\kappa} \cap I =\emptyset$
can be verified in polynomial time in $\distp$, it follows that verifying
violation of parallel-soundness can be done in polynomial time in the size of
the input, by additionally providing an appropriate node $\kappa$, instance $I$,
where $\ssize{I} \le\max_{j \in \{1,\ldots,m\}}\{\ssize{\bodypos{\qr_j}}\}+m$,
and integer $i$ denoting disjunct $\qr_i$.


    \medskip \noindent \textbf{2.}
From the definition of parallel-completeness, it follows that $\qr$ is not
parallel-complete on $\distp$ if and only if there is an instance
$I\subseteq\facts{\sch, U}$ and fact $\fc$, such
that $\fc \in \qr(I)$, and for all nodes $\kappa$, $\fc\not \in \qr(\distvar{\distp}{I}(\kappa))$.
The former particularly implies that for some valuation $V$ and integer $i\in\{1,\ldots,
m\}$,
$V(\head{\qr_i}) = \fc$, $V(\pos{\qr_i}) \subseteq I$, and $V(\negbody{\qr_i})\cap I =
\emptyset$. The latter implies that for all integers $j \in
\{1,\ldots, m\}$, the valuation $V_j$ defined by the head of $\qr_j$ and fact $\fc$,
does not satisfy on any of the nodes $\kappa$, that is, $V_j(\pos{\qr_j}) \not\subseteq
\distvar{\distp}{I}(\kappa)$ or $V_j(\negbody{\qr_j}) \cap \distvar{\distp}{I}(\kappa) \ne \emptyset$.

Again the size of $I$ can be bounded, in particular, it suffices to consider only
instances of size at most $\max_{j \in \{1, \ldots,
m\}}\{\ssize{\bodypos{\qr_j}}\}+\ssize{\nw}\cdot m$, where $\nw$ denotes the
network where $\distp$ is defined over.
Indeed, for every instance $I$ as in the condition, we can construct an instance $I'$
containing all the facts in $V(\pos{\qr_i})$ and for every node $\kappa$, and every
$j\in\{1, \ldots m\}$ either add nothing (when $V_j(\pos{\qr_j})\not\subseteq I$), or add
some fact $\fcB_j \in V_j(\negbody{\qr_j}) \cap I \cap \respvar{\distp}{\kappa}$. 

By definition, the considered distribution policies have only polynomially many nodes in
$n$, and verifying whether a fact is available on some node can be done in polynomial
 as well. Hence, the result follows.


    \medskip \noindent \textbf{3.}
    As parallel-correctness for $\qr$ under $\distp$ means $\qr$ being parallel-sound
    and parallel-complete under $\distp$, it follows immediately from $(1)$ and
    $(2)$ that deciding parallel-correctness for $\classUFCQnegineq$ and
    $\classpenum$ is in $\coNP$.
\end{proof}

\end{document}